\documentclass[11pt,oneside]{article}

\usepackage{amsthm, amsmath, amssymb, amsfonts, graphicx, epsfig, mathtools, mathdots, epstopdf}

\usepackage{bbm}
\usepackage{slashbox,rotating}
\usepackage{enumerate}

\usepackage[normalem]{ulem}

\usepackage[colorlinks=true, pdfstartview=FitV, linkcolor=blue,
            citecolor=blue, urlcolor=blue]{hyperref}
\usepackage[usenames]{color}
\definecolor{Red}{rgb}{0.7,0,0.1}
\definecolor{Green}{rgb}{0,0.7,0}

\usepackage{url}
\makeatletter\def\url@leostyle{%
 \@ifundefined{selectfont}{\def\UrlFont{\sf}}{\def\UrlFont{\scriptsize\ttfamily}}} \makeatother

\newtheorem{theorem}{Theorem}[section]
\newtheorem{proposition}[theorem]{Proposition}
\newtheorem{lemma}[theorem]{Lemma}

\theoremstyle{definition}
\newtheorem{definition}[theorem]{Definition}
\newtheorem{example}{Example}[section]
\theoremstyle{remark}
\newtheorem{remark}[theorem]{Remark}

\newtheoremstyle{dotless}{}{}{\itshape}{}{\bfseries}{}{ }{}
\theoremstyle{dotless}

\def\cC{\mathcal{C}}
\def\cD{\mathcal{D}}
\def\cE{\mathcal{E}}
\def\cF{\mathcal{F}}

\def\cH{\mathcal{H}}
\def\cI{\mathcal{I}}

\def\cL{\mathcal{L}}

\def\cP{\mathcal{P}}
\def\cQ{\mathcal{Q}}
\def\cR{\mathcal{R}}
\def\cS{\mathcal{S}}
\def\cT{\mathcal{T}}
\def\cU{\mathcal{U}}

\def\cX{\mathcal{X}}

\def\cZ{\mathcal{Z}}

\def\bE{\mathbb{E}}
\def\bF{\mathbb{F}}

\def\bN{\mathbb{N}}

\def\bP{\mathbb{P}}
\def\bQ{\mathbb{Q}}
\def\bR{\mathbb{R}}

\def\1{\mathbbm{1}}

\newcommand{\set}[1]{\left\{#1\right\}} 

\def\d{\mathrm{d}}

\title{Dynamic Conic Finance: Pricing and Hedging in Market Models with Transaction Costs via Dynamic Coherent Acceptability Indices
\\[0.5ex] }
\author{Tomasz R. Bielecki\footnote{Tomasz R. Bielecki and Igor Cialenco acknowledge support from the NSF grant DMS-0908099.}\\
\small{Department of Applied Mathematics,}\\[-0.3ex]
\small{Illinois Institute of Technology,}\\[-0.3ex]
\small{Chicago, 60616 IL, USA}\\[-0.3ex]
\url{bielecki@iit.edu}\\
\and
Igor Cialenco\footnotemark[\value{footnote}]\\[-0.3ex]
\small{Department of Applied Mathematics,}\\[-0.3ex]
\small{Illinois Institute of Technology,}\\[-0.3ex]
\small{Chicago, 60616 IL, USA}\\[-0.3ex]
\url{igor@math.iit.edu} \\
\and
Ismail Iyigunler \\
\small{Department of Applied Mathematics,}\\[-0.3ex]
\small{Illinois Institute of Technology,}\\[-0.3ex]
\small{Chicago, 60616 IL, USA}\\[-0.3ex]
\url{iiyigunl@hawk.iit.edu} \\
\and
Rodrigo Rodriguez \\
\small{Department of Applied Mathematics,}\\[-0.3ex]
\small{Illinois Institute of Technology,}\\[-0.3ex]
\small{Chicago, 60616 IL, USA}\\[-0.3ex]
\url{rrodrig8@hawk.iit.edu}}

\date{First Circulated: May 12, 2012\\
This Version: February 10, 2013\\
Forthcoming in IJTAF}

\begin{document}

\maketitle

\begin{abstract}
In this paper we present a theoretical framework for determining dynamic ask and bid prices of derivatives using the theory of dynamic coherent acceptability indices in discrete time.
We prove a version of the First Fundamental Theorem of Asset Pricing using the dynamic coherent risk measures.
We introduce the dynamic ask and bid prices of a derivative contract in markets with transaction costs.
Based on these results, we derive a representation theorem for the dynamic bid and ask prices in terms of dynamically consistent sequence of sets of probability measures and risk-neutral measures.
To illustrate our results, we compute the ask and bid prices of some path-dependent options using the dynamic Gain-Loss Ratio.
\end{abstract}

{\noindent \small
}


\section{Introduction}
We develop a  framework for narrowing the theoretical spread between ask prices and bid prices of derivative securities in discrete-time market models with transaction costs, using dynamic coherent acceptability indices (DCAIs) that are studied in Bielecki, Cialenco, and Zhang~\cite{BCZ2010}.
Aside from the use of acceptability indices as a tool, our approach is related to the literature studying no-good-deal pricing as a vehicle to narrow the no-arbitrage interval.

We first formulate and prove a no-good-deal version of the fundamental theorem of asset pricing (FTAP) using a \emph{family} of dynamic coherent risk measures associated with a DCAI.
The classic form of FTAP, i.e. the no-arbitrage form of FTAP in frictionless markets, has been established by numerous authors in varying degrees of generality (Harrison and Pliska \cite{HarrisonPliska1981}, Dalang, Morton, and Willinger \cite{DalangMortonWillinger1990}, Schachermayer \cite{Schachermayer1992}, Rogers \cite{Rogers1994}, Kabanov and Kramkov \cite{KabanovKramkov1994}, Jacod and Shiryaev \cite{JacodShiryaev1998}, Kabanov and Stricker \cite{KabanovStricker2001}); for continuous time see Delbaen and Schachermayer \cite{DelbaenSchachermayer1994,DelbaenSchachermayer1996}, Cherny~\cite{Cherny2007e}).
For markets with transaction costs, no-arbitrage versions of the FTAP are proved in Jouini and Kallal~\cite{JouiniKallal1995},  Kabanov and Stricker~\cite{KabanovStricker}, Kabanov, R\'{a}sonyi, and Striker~\cite{KabanovRasonyiStricker2002}, Schachermayer~\cite{Schachermayer2004}, and Bielecki, Cialenco, and Rodriguez~\cite{BieleckiCialencoRodriguez2012}.
In Carr, Geman, and Madan~\cite{CarGemanMadan2001}, the FTAP was formulated and proved in terms of the no strictly acceptable opportunities condition for frictionless markets, and subsequently  Pinar, Salih, and Camci~\cite{PinarSalihCamcl2010} proved a version of the FTAP in the context of the Gain-Loss ratio in markets with proportional transaction costs.
The no-good-deal version of FTAP has been obtained for markets with transaction costs in the context of static coherent risk measures, and for frictionless markets using discrete-time coherent risk measures by Cherny~\cite{Cherny2007} and \cite{Cherny2007c}, respectively.

There is an extensive literature for methods that narrow the theoretical no-arbitrage interval.
One of the widely studied approaches is indifference pricing, which is based on utility maximization.
Specifically, an indifference price is a price at which an agent receives the same expected utility between trading and not trading.
A comprehensive collection of articles related to indifference pricing can be found in Carmona~\cite{CarmonaIndifferencePricing2009}.
However, it is known that the indifference pricing approach has limitations: numerical implementations and explicit calculations for indifference pricing may not be robust, and the resulting bid and ask prices are not necessarily risk-neutral in practice (see for instance Staum~\cite{Staum2007}).
Alternatively, Cochrane and Sa\'{a}-Requejo~\cite{Cochrane2000} introduced the no-good-deal pricing methodology.
In this approach, the arbitrage bounds are narrowed by ruling out deals that are too good---cash flows that have high Sharpe ratios.
This strengthens the no-arbitrage argument by assuming that any investor is willing to accept a good-deal.
In a subsequent papers by Bernardo and Ledoit~\cite{Bernardo2000} and Pinar, Salih, and Camci~\cite{PinarSalihCamcl2010} cash flows are considered good deals if their corresponding  Gain-Loss ratio is high.
The no-good-deal pricing approach has been used in other applications and settings by Carr, Geman, and Madan~\cite{CarGemanMadan2001}, Jaschke and Kuchler~\cite{JaschkeKuchler2001}, Staum~\cite{Staum2004}, Roorda, Schumacher, and Engwerda~\cite{Roorda2005a}, Bj\"{o}rk and Slinko~\cite{Bjork2006}, Kloppel and Schweitzer~\cite{KloppelSchweitzer2007}, Arai and Fukasawa \cite{AraiFukasawa2011}.
The no-good-deal pricing has also been approached via coherent risk measures in Cherny and Madan~\cite{Cherny2006c} and Cherny~\cite{Cherny2007c}.

Several authors studied no-good-deal pricing with either discrete-time or continuous time risk measures.
In Madan, Pistorius, and Schoutens~\cite{MadanPisSch2011}, dynamically consistent bid and ask prices for structured products are derived using nonlinear expectations, and in Bion-Nadal~\cite{BionNadal2009a} and Cherny~\cite{Cherny2007} dynamic bid and ask prices are found via dynamic risk measures.

Cherny and Madan~\cite{MadanCherny2010} proposed the conic finance framework for pricing in incomplete, frictionless markets using static acceptability indices, which are introduced in Cherny and Madan~\cite{ChernyMadan2009}.
The framework is called conic finance because the derivative prices they introduce depend on the direction of trade---the resulting set of cash flows generated by the prices of the derivative is longer a linear space, it is instead a \emph{convex cone}.
However, as with any static pricing technique, their prices may lack a dynamic consistency property.
This drawback renders the static approach inadequate for pricing exotic derivatives such as path-dependent derivatives.
In a  recent study,  Rosazza-Gianin and Sgarra ~\cite{RosazzaGianinSgarra2012} apply  the concepts of dynamic acceptability indices and of $g$-expectation to investigate liquidity risk.

Compared to the papers above, our contributions amount to the following:
\begin{itemize}
\item
Our framework allows for (hedging) cash flows to pay dividends, and be subjected to transaction costs.
In particular, we can apply our no-good-deal pricing approach to the pricing of interest rate swaps and credit default swaps in markets with transaction costs.
\item
We prove a version of the FTAP formulated in terms of a no-good-deal condition.
It is important to  stress that our no-good-deal condition is dynamically consistent in time.
\item
We construct the good-deal ask and bid prices of a derivative which are dynamically consistent, in the sense that they are defined in terms of dynamic coherent acceptability indices.
This allows us to narrow the no-arbitrage pricing interval.
\item
We exemplify the proposed general theory with the dynamic Gain-Loss ratio, which is a particular dynamic coherent acceptability index.
\end{itemize}

This paper is organized as follows.
In Section~\ref{sec:arbANDgooddeals}, we define the no-arbitrage condition and the no-good-deal condition, and then prove the Fundamental Theorem of Good-Deal Pricing.
Next, in Section~\ref{sec:DynamicBidAskviaDCAI}, we define the no-good-deal ask and bid prices, and proceed by proving a representation theorem for them.
Finally, in Section~\ref{sec:PricingDGLR}, we use dynamic Gain-Loss Ratio  to compute the good-deal ask and bid prices for some path-dependent, European-style options.

\section{Arbitrage and good-deals}\label{sec:arbANDgooddeals}

We extensively use the results on dynamic acceptability indices that were obtained in \cite{BCZ2010}.
Thus, we adopt  the mathematical set-up that was used therein.
In particular, we assume that the underlying probability space is finite, an assumption that indeed is made so to simplify the presentation.

Let $T$ be a fixed time horizon, and let $\cT:=\{0, 1, \dots, T\}$. Next, let $(\Omega, \cF_T, \mathbb{F}=(\cF_t)_{t\in \cT}, \bP)$ be the underlying  filtered probability space, and assume that $\Omega=\{\omega_1, \dots, \omega_N\}$, and $\bP$ is of full support.
In what follows, we will denote by $L^0:=L^0(\Omega, \cF_T, \bF, \bP)$ the set of all $\bF$-adapted processes.

On this probability space, we consider a market consisting of a savings account $B$ and of $N$ traded securities satisfying the following properties:
\begin{itemize}
\item
The savings account can be purchased and sold according to the process \\ $B:=\big((\prod_{s=0}^t[1+r_s])\big)_{t=0}^T$, where $(r_t)_{t=0}^T$ is  a nonnegative process specifying the risk-free rate.
\item
 The $N$ securities can be purchased according to the ex-dividend price process \\
 $P^{ask}:=\big((P^{ask,1}_t, \dots, P^{ask, N}_t)\big)_{t=0}^T$; the associated (cumulative) dividend process is denoted by $A^{ask}:=\big((A^{ask, 1}_t, \dots, A^{ask, N}_t)\big)_{t=1}^T$.
\item
The $N$ securities  can be sold  according to the ex-dividend price process \\
$P^{bid}:=\big((P^{bid,1}_t, \dots, P^{bid, N}_t)\big)_{t =0}^T$; the associated (cumulative) dividend process is denoted by $A^{bid}:=\big((A^{bid, 1}_t, \dots, A^{bid, N}_t)\big)_{t=1}^T$.
\end{itemize}
\noindent
We assume that the processes introduced above are adapted.
Unless stated otherwise, all inequalities and equalities involving vector-valued processes are understood coordinate-wise.
In what follows, we shall denote by $\Delta $ the backward difference operator: $\Delta X_t:=X_t-X_{t-1}$, and we take the convention that $A^{ask}_0=A^{bid}_0=0$.

\begin{remark}
For any $t =1, 2, \dots, T$ and $j=1, 2, \dots, N$, the random variable $\Delta A^{ask, j}_t$ is interpreted as amount of dividend associated with holding a \emph{long} position in security $j$ from time $t-1$ to time $t$. Respectively,  the random variable $\Delta A^{bid, j}_t$ is interpreted as amount of dividend associated with holding a \emph{short} position in security $j$ from time $t-1$ to time $t$.
\end{remark}

Let us illustrate the processes introduced above in the context of a Credit Default Swap  (CDS) contract.
\begin{example} A CDS contract is a contract between two parties, a \emph{protection buyer} and a \emph{protection seller}, in which the protection buyer pays periodic fees to the protection seller in exchange for some payment made by the protection seller to the protection buyer if a pre-specified credit event of a reference entity occurs.
Let $\tau$ be the nonnegative random variable specifying the time of the credit event of the reference entity.
Suppose the CDS contract admits the following specifications: initiation date $t=0$, expiration date $t=T$, nominal value \$1, and the loss-given-default is given by a nonnegative scalar $\delta$ and is paid at default.
Typically, CDS contracts are traded on over-the-counter markets in which dealers quote CDS spreads to investors.
Suppose that the CDS spread quoted by the dealer to sell a CDS contract is $\kappa^{bid}$, and the CDS spread quoted by the dealer to buy a CDS contract is $\kappa^{ask}$.
For the CDS contract specified above, the cumulative dividend processes $A^{ask}$ and $A^{bid}$  are defined as follows
\begin{align*}
 A^{ask}_t&:=1_{\{\tau \leq t\}} \delta -\kappa^{ask} \sum_{u=1}^t 1_{\{u <\tau\}} \quad \text{and} \quad  A^{bid}_t:=1_{\{\tau \leq t\}} \delta -\kappa^{bid} \sum_{u=1}^t 1_{\{u <\tau\}}
\end{align*}
for $t \in \cT$.
In this case, the ex-dividend ask and bid price processes $P^{bid}$ and $P^{ask}$ specify the mark-to-market values of the CDS for the protection seller and protection buyer, respectively, from the perspective of the protection buyer.
\end{example}

From now on, we make the following natural standing assumption.

\medskip

\noindent {\bf Assumption (A):}
$\ P^{ask} \geq P^{bid}$ \textup{and} $\Delta  A^{ask} \leq \Delta  A^{bid}$.

\medskip

\subsection{Self-financing trading strategies}\label{subsec:Self-financing}

A \emph{trading strategy} is a predictable process $\phi:=\big((\phi^0_t, \phi^1_t, \dots, \phi^N_t)\big)_{t=1}^T$, where $\phi^j_t$ is interpreted as the number of units of security $j$ held from time $t-1$ to time $t$.
We take the convention that $\phi^0$ corresponds to the holdings in the savings account $B$, and  $\phi_0=(0,\ldots,0)$.

\begin{definition}\label{Def:WealthProcess}
The \emph{wealth process} $V(\phi)$ associated with a trading strategy $\phi$ is defined as
\[V_t(\phi) = \left\{
\begin{array}{l l}
\phi^{0}_{1}+\sum_{j=1}^N \1_{\{\phi^j_{1} \geq 0\}}\phi^{j}_{1}P^{ask, j}_0+\sum_{j=1}^N  \1_{\{\phi^j_{1} < 0\}}\phi^{ j}_{1}P^{bid, j}_0, &\mbox{if $t=0$},\\[.05in]
\phi^{0}_{t}B_{t}+\sum_{j=1}^N \1_{\{\phi^j_{t} \geq 0\}}\phi^{j}_{t}(P^{bid, j}_t+\Delta A^{ask, j}_t)\\
 \qquad+\sum_{j=1}^N  \1_{\{\phi^j_{t} < 0\}}\phi^{ j}_{t}(P^{ask, j}_t+\Delta A^{bid, j}_t), &\mbox{if $1\leq t\leq T$}.\\[.05in]
  \end{array} \right. \]
\end{definition}

\begin{remark}
(i)  It is important to note the difference in the use of bid and ask prices, in the above definition, between the time $t=0$ and the time $t\in\set{1, \ldots, T}$.
At time $t=0$, $V_0(\phi)$ is interpreted as the cost of setting up the portfolio associated with $\phi$.
   For $t=1, \ldots, T$, the wealth process $V_t(\phi)$ equals the sum of the liquidation value of the portfolio associated with trading strategy $\phi$ before any time $t$ transactions and the dividends associated with $\phi$ from time $t-1$ to $t$.\\
   (ii) Also note that, due to the presence of transaction costs, the wealth process $V$ may not be linear in its argument, i.e. $V(\phi)+V(\psi)\neq V(\phi+\psi)$, and $V(\alpha\phi)\neq \alpha V(\phi)$ for $\alpha \in\bR$,  and some trading strategies $\phi, \psi$.
This is the major difference from the frictionless setting.
\end{remark}

We proceed by introducing the self-financing condition, which is appropriate in the context of this paper.

\begin{definition}\label{Def:SelfFinancing}
A trading strategy $\phi$ is self-financing if
\begin{align}\label{Eq: Phi}
B_{t}\Delta  \phi^{0}_{t+1}+\sum_{j=1}^NP^{ask, j}_{t}\1_{\{\Delta  \phi^j_{t+1} \geq 0\}} \Delta  \phi^{ j}_{t+1}+\sum_{j=1}^NP^{bid, j}_{t}\1_{\{\Delta  \phi^j_{t+1} < 0\}} \Delta  \phi^{ j}_{t+1}\\ \notag
  \qquad \qquad =  \sum_{j=1}^N \phi^{ j}_{t}\1_{\{ \phi^j_{t} \geq 0\}} \Delta  A^{ask, j}_{t}+\sum_{j=1}^N \phi^{j}_{t}\1_{\{ \phi^j_{t} < 0\}} \Delta  A^{bid, j}_{t}\notag
  \end{align}
  for all $t=1, 2,  \ldots, T-1$.
\end{definition}
The self-financing condition guarantees that no money can flow in or out of the portfolio.

In what follows, we shall work with the discounted processes:  $V^*(\phi):=B^{-1}V(\phi)$ for all trading strategies $\phi$.
The next result gives a useful characterization of the self-financing condition in terms of the wealth process.
For the proof we refer to Bielecki, Cialenco, and Rodriguez~\cite{BieleckiCialencoRodriguez2012}.

\begin{lemma}\label{Lemma:DiscountedSelfFinancing}
A trading strategy $\phi$ is self-financing if and only if the wealth process $V(\phi)$ satisfies the following equality
\begin{align*}
&V^*_t(\phi)=V_0(\phi)+\sum_{j=1}^N \1_{\{\phi^j_t \geq 0\}}\phi^{j}_t B^{-1}_tP^{bid, j}_t
+\sum_{j=1}^N \1_{\{\phi^j_t < 0\}}\phi^{ j}_tB^{-1}_tP^{ask, j}_t\\
& \quad  -\sum_{j=1}^N\sum_{u=1}^t \1_{\{\Delta \phi^j_u \geq 0\}}\Delta \phi^j_{u}B^{-1}_{u-1}P^{ask, j}_{u-1}
-\sum_{j=1}^N\sum_{u=1}^t\1_{\{\Delta \phi^j_u < 0\}}\Delta \phi^j_u B^{-1}_{u-1}P^{bid, j}_{u-1}\\
& \quad +\sum_{j=1}^N \sum_{u=1}^t\1_{\{\phi^j_u \geq 0\}}\phi^{ j}_uB^{-1}_u \Delta  A^{ask,j}_u
+\sum_{j=1}^N\sum_{u=1}^t\1_{\{\phi^j_u < 0\}}\phi^{ j}_uB^{-1}_u \Delta A^{bid,j}_u
\end{align*}
for $t =1, 2, \dots, T$.
\end{lemma}

\noindent Thus, the wealth process at time $t$, associated with a self-financing trading strategy $\phi$, is equal to the sum of setting up the portfolio associated with $\phi$ at time $t=0$, the liquidation value at time $t$ of the portfolio associated with $\phi$, all purchases and sales before time $t$, and all dividends associated with $\phi$ up to time $t$.

\begin{remark}
Naturally, if there are no transactions costs, we recover classic definitions of the wealth process and self-financing condition.
In the case when the market is frictionless and there are no dividend-paying securities, that is $P^{ask}=P^{bid}$ and $A^{ask}=A^{bid}=0$,
see for instance Pliska~\cite{Pliska}.
If  the market is frictionless and there are dividend-paying securities, that is $P^{ask}=P^{bid}$ and $A^{ask}=A^{bid}$, see for example Kijima~\cite{Kijima2003}.
\end{remark}

\subsection{Arbitrage}\label{subsec:Arbitrage}
We start with defining the following sets of self-financing trading strategies.
\begin{align*}
\cS(t) & :=
\begin{cases}
\{ \phi : \phi \; \text{ is s.f.}, V_0(\phi)=0 \} , & t=0 \nonumber \\
\{ \phi : \phi \; \textrm{ is s.f., } \phi_s=\1_{\{s \geq t+1\}}\phi_s \textrm{ for all }  s=1, 2, \dots, T\}, & t\in \{1, \dots, T-1\} \nonumber
\end{cases}
\end{align*}
Note that in particular $V_t(\phi)=0$ for any $\phi \in \cS(t).$ \\
Also, we define
\begin{align}
  \cH^0(t) &:= \Big\{ \Big(0,\dots, 0, \Delta V^*_{t+1}(\phi), \ldots, \Delta V^*_{T}(\phi) \Big) \,:\, \phi\in \cS(t)\Big\} \label{Eq:DefUnderlyingMarket0}
\end{align}
for $t\in \{0, \dots, T-1\}$.
We  call  $\cH^0(t)$ the set of \emph{hedging cash flows initiated at time $t$}.

Due to the presence of transaction costs, the sets $\cH^0(t)$, generally speaking, are not convex, and for this reason we introduce the following auxiliary sets.
\begin{align}
  \cL_+(t) & := \Big\{ (Z_s)_{s=0}^T  :    Z_s\in L_+(\Omega,\cF_s,\bP), \  Z_s=\1_{\{s \geq t+1\}}Z_s, s=0,\ldots, T\Big\},   \\
\cH(t) & : = \Big\{ \Big(0,\dots, 0, \Delta (V^*_{t+1}(\phi)-Z_{t+1}), \ldots, \Delta ( V^*_{T}(\phi) - Z_{T}) \Big) \,:\,
\phi\in \cS(t), \ Z\in\cL_+(t)\Big\}, \label{Eq:DefUnderlyingMarket}
\end{align}
for $t\in \{0, \dots, T-1\}$.
We  will also refer to $\cH(t)$ as the \emph{set of hedging cash flows initiated at time $t$}.
Moreover, using the fact that the set
$
   \{V^*_{s}(\phi)-X: \phi \; \text{is s.f.},  \; X \textrm{ is }  \cF_s-\textrm{measurable, and } X\geq0 \}
$   is a convex cone (see \cite{BieleckiCialencoRodriguez2012}), it is easy to show that the set $\cH(t)$ is also a convex cone.

Let us proceed by defining an arbitrage opportunity in our setting.

\begin{definition}\label{Def:ArbOpp}
An \emph{arbitrage opportunity} at time $t\in \{0, \dots, T-1\}$ for $\cH^0(t)$ is a cash flow $H \in \cH^0(t)$ such that $\sum_{s=t+1}^TH_s(\omega) \geq 0$ for all $\omega \in \Omega$, and $\bE^{\bP}_t [\sum_{s=t+1}^T H_s ](\omega)>0$ for some $\omega \in \Omega$.
\end{definition}
We say that the \emph{no-arbitrage condition} holds true at time $t \in \{0, \dots, T-1\}$ for $\cH^0(t)$ if there does not exist an arbitrage opportunity at time $t \in \{0, \dots, T-1\}$ for $\cH^0(t)$.
\begin{remark}
Typically, arbitrage is defined as a trading strategy rather than a cash flow.
However, in our setting, it is more convenient to work with cash flows, and since each hedging cash flow corresponds to a trading strategy, we take the liberty to define an arbitrage opportunity as a cash flow.
\end{remark}

\begin{definition}\label{Def:RiskNeutralMeasure0}
For any fixed $t \in \{0, \dots, T-1\}$, we say that a probability measure $\bQ$ is \emph{risk-neutral for} $\cH^0(t)$ if $\bQ \sim \bP$, and if $\bE^{\bQ}_t[\sum_{s=t+1}^TH_s](\omega)\leq 0$ for all $\omega \in \Omega$ and all $H \in \cH^0(t)$.
The set of all risk-neutral measures for $\cH^0(t)$ will be denoted by  $\cR(\cH^0(t)).$
\end{definition}

Similarly to the above, we define the set $\cR(\cH(t))$ of risk-neutral probabilities, the arbitrage opportunity for set $\cH(t)$, and no-arbitrage conditions for set $\cH(t), \ t \in \{0, \dots, T-1\}$.
The following two lemmas show that we may formally replace $\cH^0(t)$ by $\cH(t)$ in Definitions~\ref{Def:ArbOpp} and~\ref{Def:RiskNeutralMeasure0}.

\begin{lemma}\label{lemma: RHandRH0}
For any $t\in \{0, \dots, T-1\}$,  we have that $\bQ \in \cR(\cH^0(t))$ if and only if $\bQ \sim \bP$, and $\bE^{\bQ}_t[\sum_{s=t+1}^TH_s]\leq 0$ for all $H \in \cH(t)$.

\end{lemma}

\begin{proof}
Fix $t \in \{0, \dots, T-1\}$.

\noindent
  $(\Longrightarrow) $     If $\bQ \in \cR(\cH^0(t))$, then $\bE^{\bQ}_t\big[\sum_{s=t+1}^TH^0_s\big]\leq 0$ for all $H^0 \in \cH^0(t)$.
    Hence, \\ $\bE^{\bQ}_t\big[\sum_{s=t+1}^TH^0_s-Z_T\big]\leq 0$
    for all $H^0 \in \cH^0(t)$ and $Z \in \cL_+(t)$.
  Therefore,
  $\bE^{\bQ}_t[\sum_{s=t+1}^T H_s]\leq 0$ for all $H \in \cH(t)$.

 \vspace{10pt}

\noindent
  $(\Longleftarrow)$    Suppose that $\bQ \sim \bP$, and that $\bE^{\bQ}_t[\sum_{s=t+1}^TH_s]\leq 0$ for  all $H \in \cH(t)$.\\
  Then,  $\bE^{\bQ}_t\big[\sum_{s=t+1}^TH^0_s-Z_T\big]\leq 0$ for all $H^0 \in  \cH^0(t)$ and $Z \in \cL_+(t)$.
  Letting $Z_T=0$ proves that $\bQ \in \cR(\cH^0(t))$.
  \end{proof}

\begin{lemma}\label{lemma: HandH0}
For each $t \in \{0, \dots, T-1\}$, the no-arbitrage condition holds true at time $t$ for $\cH^0(t)$ if and only if for each  $H \in \cH(t)$ such that $\sum_{s=t+1}^T H_s \geq0$, we have $\sum_{s=t+1}^T H_s =0$.
\end{lemma}

\begin{proof}
Let us fix $t \in \{0, \dots, T-1\}$.

\noindent
$(\Longrightarrow)$
Assume that $H \in \cH(t)$ is such that $\sum_{s=t+1}^T H_s \geq0$.
Then, by definition of $\cH(t)$, there exists $H^0 \in \cH^0(t)$ and $Z \in \cL_+(t)$ so that $\sum_{s=t+1}^T H_s=\sum_{s=t+1}^T H^0_s-Z_T$.
This give us $\sum_{s=t+1}^T H^0_s\geq Z_T$.
The no-arbitrage condition holds true at time $t$ for $\cH^0(t)$, so $\sum_{s=t+1}^T H^0_s= 0$.
Therefore, $Z_T=0$, which implies $\sum_{s=t+1}^T H_s =0$.

\noindent
$(\Longleftarrow) $
Suppose that $H^0 \in \cH^0(t)$ is such that $\sum_{s=t+1}^T H^0_s \geq0$.
By assumption, for each  $H \in \cH(t)$ such that $\sum_{s=t+1}^T H_s \geq0$, we have $\sum_{s=t+1}^T H_s =0$.
From the definition of $\cH(t)$, this implies that for each  $\hat H^0 \in \cH^0(t)$, $Z \in \cL_+(t)$ such that $\sum_{s=t+1}^T \hat H_s-Z_T \geq0$, we have $\sum_{s=t+1}^T\hat H^0_s-Z =0$.
Taking $Z=0$ and $\hat H^0:=H^0$ gives us $\sum_{s=t+1}^T H^0_s =0$.
\end{proof}

\smallskip
In what follows we shall make use of the following result.

\begin{proposition}\label{theorem:FFTAP}
If $\cR(\cH(t)) \neq \emptyset$, then the no-arbitrage condition holds at time $t \in \{0, \dots, T-1\}$ for $\cH(t)$.
\end{proposition}

\begin{proof}
We prove by contradiction.
Assume that $\bQ \in \cR(\cH(t))$, and that there exists an arbitrage opportunity $H$ at time $t \in \{0, \dots, T-1\}$.
By the definition of an arbitrage opportunity, $H \in \cH(t)$, $\sum_{s=t+1}^T H_s \geq0$, and $\bE^{\bP}_t\big[\sum_{s=t+1}^T H_s](\omega) > 0$ for some $\omega \in \Omega$.
  Since $\bQ \sim \bP$ and $\sum_{s=t+1}^T H_s\geq0$, we have that $\bE^{\bQ}_t\big[\sum_{s=t+1}^T H_s](\omega) > 0$ for some $\omega \in \Omega$.
However, this contradicts that $\bQ \in \cR(\cH(t))$.
Hence, the no-arbitrage condition holds true at time $t \in \{0, \dots, T-1\}$ for $\cH(t)$.
\end{proof}

Next, we introduce some notions that are related to derivatives pricing, and which will be used in Section~\ref{sec:PricesRep}.
In what follows, for any cash flow $D \in L^0$ we will denote by $D^*:=B^{-1}D$ the discounted cash flow.

\begin{definition}
Let $t \in \{0, \dots, T-1\}$.
\begin{itemize}
\item
\emph{A set of extended cash flows associated with an $\cF_t$-measurable random variable $S_t$ and a process $D \in L^0$} is defined as
\begin{align*}
\widetilde \cH(t, S_t) :=&\Big\{ \Big(0, \dots,0,\xi_t S_t, H_{t+1}-\xi_t D^*_{t+1}, \dots, H_T-\xi_t D^*_T\Big)\nonumber \\
& \qquad : H \in \cH(t),\; \xi_t \;\text{is an}\; \cF_t\text{-measurable r.v.}\Big\},
\end{align*}
\item
\emph{The pricing interval associated with  a process $D \in L^0$ and a set of probability measures $\cX$} is defined as
\begin{align*}
\cI(t, D; \cX)&:=  \Big\{\bE^{\mathbb{Q}}_t \Big [ \sum _{s=t+1}^T D^*_s \Big]: \mathbb{Q} \in \cX\Big\}.
\end{align*}
\end{itemize}
\end{definition}
A cash flow in $\widetilde\cH(t, S_t)$ is interpreted as the sum of a position in $\cH(t)$ and a \emph{static} position of $\xi_t$ units in the discounted cash flow $(0, \dots, 0,  S_t, - D^*_{t+1}, \dots, -D^*_T)$.
In Section~\ref{sec:PricesRep}, $S_t$ will have the interpretation of a discounted price of the cash flow $D$.

We will say that $\cI(t, D; \cX)$ is a \emph{risk-neutral pricing interval} if it is nonempty, and if for each $S_t \in \cI(t, D;  \cX)$ the no-arbitrage condition is satisfied for $\widetilde \cH(t, S_t)$.
That is, $\cI(t, D; \cX)$ is a risk-neutral pricing interval if it is nonempty, and if  for each $S_t \in \cI(t, D; \cX)$ and each $\widetilde H \in \widetilde \cH(t, S_t)$ such that $\sum_{s=t+1}^T \widetilde H_s\geq0$, we have $\sum_{s=t+1}^T\widetilde H_s=0$.
If $\cI(t, D; \cX)$ is a risk-neutral pricing interval, we call any $S_t \in \cI(t, D; \cX)$ a \emph{risk-neutral price},
$\sup_{ \bQ \in \cR(\cH(t))} \bE^\bQ_t\big[\sum_{s=t+1}^TD^*_s\big]$ \emph{the upper no-arbitrage bound}, and $\inf_{ \bQ \in \cR(\cH(t))} \bE^\bQ_t\big[\sum_{s=t+1}^TD^*_s\big]$ \emph{the lower no-arbitrage bound}.

The following lemma gives a necessary condition for $\cI(t, D, \cX)$ to be a risk-neutral pricing interval.

\begin{lemma}\label{lemma:NAinterval}
Let $t \in \{0, \dots, T-1\}$ and $D \in L^0$.
If $\cR(\cH(t)) \neq \emptyset$, then $ \cI(t, D; \cR(\cH(t)))$ is a risk-neutral pricing interval.
\end{lemma}

\begin{proof}
  Fix  $t \in \{0, \dots, T-1\}$,  $D \in L^0$, and $S_t \in \cI(t, D;\cR(\cH(t)))$.
Let $\widetilde H \in \widetilde \cH(t, S_t)$ be a cash flow such that $\sum_{s=t+1}^T\widetilde H_s \geq0$.
By definition of $\widetilde \cH(t, S_t)$, we have that
\begin{equation}\label{eq:NAInterval}
\xi _t S_t+\sum_{s=t+1}^T (H_s -\xi_tD^*_s) \geq 0
\end{equation}
 for some $H \in \cH(t)$ and some $\cF_t$-measurable random variable $\xi_t$.

Now, since $\cR(\cH(t)) \neq \emptyset$ and $S_t \in \cI(t, D;\cR(\cH(t)))$, there exists $\bQ \in \cR(\cH(t))$ such that $S_t=\bE^\bQ_t\big[\sum_{s=t+1}^TD^*_s\big]$.
It follows that $\xi_t \bE^\bQ_t\big[\sum_{s=t+1}^TD^*_s\big]-\xi_tS_t=0$.
From \eqref{eq:NAInterval} we see that $\bE^\bQ_t\big[\sum_{s=t+1}^T H_s\big] \geq 0$ holds.
Since $\bQ \in \cR(\cH(t))$, we have   that $\bE^\bQ_t\big[\sum_{s=t+1}^T H_s\big] = 0$, which gives us that
\begin{equation}\label{eq:NAInterval2}
 \xi _t S_t+\bE^\bQ_t\Big[\sum_{s=t+1}^T (H_s -\xi_tD^*_s)\Big] = 0.
\end{equation}
From \eqref{eq:NAInterval} and \eqref{eq:NAInterval2}  we conclude that $\xi _t S_t+\sum_{s=t+1}^T (H_s -\xi_tD^*_s)=0$, which consequently implies that the no-arbitrage condition holds for $\widetilde \cH(t, S_t)$.
\end{proof}

\subsection{Good-deals}\label{subsec:GoodDeals}
The no-good-deal bound pricing approach was introduced in \cite{Cochrane2000}.
This approach assumes that all investors are willing to invest in good deals -- trades with high Sharpe ratios -- as well as in the arbitrage opportunities, if any.
In \cite{Cherny2007, Cherny2007c}, alternative approach to no-good-deal bounds was proposed: these authors suggested using coherent risk measures instead of the Sharpe Ratio.
Recently, the notion of conic finance was introduced in \cite{MadanCherny2010}, where a good deal was defined in terms of a family of static coherent risk measures.
In the present paper, we extend conic finance to a dynamic setting by defining a good deal in terms of a family of \emph{dynamic} coherent risk measures (DCAI).

The main tool for building up the theory of Dynamic Conic Finance will be the Dynamic Coherent Acceptability Indices (DCAIs) developed in~\cite{BCZ2010}. As it was shown in~\cite{BCZ2010} that any DCAI $\alpha$ can be associated with a left-continuous,  increasing  family of Dynamic Coherent Risk Measures (DCRMs) $(\rho^{\gamma})_{\{\gamma\in (0, \infty)\}}$, and consequently to a family of dynamically consistent sequences of sets of probability measures (see Appendix~\ref{sec:PricingDCAI} for definitions and related results.)
In what follows, we fix such a normalized$\, $\footnote{A DCAI is said to be normalized if for every $t \in \cT$ and $\omega \in \Omega$, there exist two portfolios $D, D' \in \cD$ so that $\alpha_t(D)(\omega)=+\infty$ and $\alpha_t(D')(\omega)=0$.}  DCAI $\alpha$, and denote by $(\rho^{\gamma})_{\{\gamma\in (0, \infty)\}}$ the corresponding family of DCRMs, and  by ${\cQ}=\left (\big(\cQ^{\gamma}_t\big)_{t \in \cT}\right)_{\gamma \in (0, \infty)}$ the corresponding family of dynamically consistent sequences of sets of probability measures.

\begin{definition}\label{def:GoodDeal}
 A \emph{good-deal} for $\cH(t)$ at time $t \in \{0, \dots, T-1\}$ and level $\gamma>0$ is a cash flow $H \in \cH(t)$ such that $\rho^{\gamma}_t(H)(\omega)<0$ for some $\omega \in \Omega$.
\end{definition}

Note that a good-deal depends on the family of DCRMs and the level $\gamma$.
A cash flow that is a good-deal with respect to a family of DCRMs might not be a good-deal with respect to another family of DCRMs. Also, note that, for a fixed family of DCRMs, a cash flow  that is a good-deal at level $\gamma_0$ might not be a good-deal at some other level $\gamma' >\gamma_0$. Although, since $\rho^\gamma$ is monotone increasing in $\gamma$, if a cash flow is a good-deal for $\gamma_0$, then it will also be a good deal for any level $\gamma'\leq \gamma_0$. We will also show later that good-deals can be described in terms of the acceptability index associated to family $\rho^\gamma$.

\begin{definition}\label{def:NGDBcond}
We say that the \emph{no-good-deal condition} (NGD) holds  for $\cH(t)$ at time $t \in \{0, \dots, T-1\}$ and level $\gamma>0$ if $\rho^{\gamma}_t(H)(\omega)\geq0$ for all $H \in \cH(t)$ and $\omega \in \Omega$.
\end{definition}

We will make the following technical assumption on  $\cQ$.\\
\noindent {\bf Assumption (B):}\label{Assumption:Density}
For each $\gamma>0$ and $t \in \cT$, any probability measure $\bQ\in\cQ_t^\gamma$ is equivalent to $\bP$, and the set
\begin{equation*}
  \cE^{\gamma}_t:= \left\{ \frac{ \mathrm{d} \bQ}{\mathrm{d}\bP} \,:  \, \bQ \in \cQ^{\gamma}_t\right \}
\end{equation*}
is closed and convex.

\smallskip
Since $\Omega$ is finite and $\bP$ is of full support, the set $\cE^{\gamma}_t$ is bounded.
Hence, $\cE^{\gamma}_t$ is compact for all $\gamma>0$ and $t \in \cT$.
In Section~\ref{sec:PricingDGLR}, we show that a family of densities $\cE$ corresponding to the dynamic Gain-Loss Ratio satisfies this assumption.

Next, we will prove one of the main results of this paper, which is analogous to FTAP.
\begin{theorem}\label{theorem:NGDFTAP}
 The no-good-deal condition \emph{(NGD)} holds true for $\cH(t)$ at time \\ $t \in \{0, \dots, T-1\}$ and level $\gamma>0$ if and only if $\cR(\cH(t)) \cap \cQ^{\gamma}_t \neq \emptyset$.
\end{theorem}

\begin{proof}
Throughout the proof we fix $t \in \{0, \dots, T-1\}$ and $\gamma>0$.

\noindent
($\Longleftarrow$)
Suppose that $\bQ^* \in\cR(\cH(t)) \cap \cQ^{\gamma}_t $.
Since $\bQ^*$ is risk-neutral, it follows that\\ $\bE^{\bQ^*}_t[\sum_{s=t+1}^T H_s]\leq0$ for all $H \in \cH(t)$.
Due to Theorem~\ref{theorem:DCRMrep} (Robust Representation  of DCRM), we have
 \begin{align*}
  -\rho^{\gamma}_t(H) &=\underset{\bQ \in \cQ^{\gamma}_{t}}{\inf} \bE^{\bQ}_t \Big[ \sum_{s=t+1}^TH_s \Big] \leq  \bE^{\bQ^*}_t \Big[ \sum_{s=t+1}^TH_s \Big]\leq 0\, ,
  \end{align*}
for all $\;H \in \cH(t)$.
Thus, $\rho^{\gamma}_t(H) \geq 0$ for any $H \in \cH(t)$, and hence NGD holds true for $\cH(t)$ at time $t $ and level $\gamma$.

\noindent
($\Longrightarrow$)
Fix $\ M\in\bN$ and $\textbf{H}:=(H^1, H^2, \dots, H^M) \in \cH(t)\times \cH(t)\times \cdots \times \cH(t)$.
Let $\cE^{\gamma}_t$ be the set defined in Assumption~(B), and let us consider the following set of matrices

\begin{equation*}
 \cZ_t(\textbf{H}):=\left\{\Bigg[\bE^{\bP}_t \Big[\eta\sum_{s=t+1}^T H^i_s \Big](\omega_j) \Bigg]_{j=1, \dots, N;i=1, \dots, M}:  \eta \in \cE^{\gamma}_t \right\}\subset \bR ^{N\times M}.
 \end{equation*}

Since $\cE^{\gamma}_t$ is compact, by continuity of  the mapping
\begin{equation*}
 \cE^{\gamma}_t \ni \eta \mapsto \bE^{\bP}_t\Big[\eta\sum_{s=t+1}^T H_s^i\Big](\omega_j), \quad \; i=1,2, \dots, M; \;j=1, 2, \dots, N,
\end{equation*}
we conclude that $\cZ_t(\textbf{H})$ is compact in $\bR^{N\times M}$.
Also note that, by convexity of $\cE^{\gamma}_t$ and linearity of conditional expectations above w.r.t. $\eta$, the set $ \cZ_t(\textbf{H})$ is convex.

Let us now define a closed and convex set $\cC:=(-\infty, 0]^{N \times M}\subseteq \bR^{N \times M}$.
We will prove by contradiction that $\cZ_t(\textbf{H}) \cap \cC \neq \emptyset$. Towards this end let us
assume that $\cZ_t(\textbf{H}) \cap \cC = \emptyset$.
By a version of Hahn-Banach theorem  (see Theorem~\ref{th:AppendHahnBanach}), there exists a linear functional $\varphi^{t, \textbf{H}}: \bR^{N \times M} \rightarrow \bR$, and $\epsilon_{t, \textbf{H}}>0$ such that
\begin{align}
   \epsilon_{t, \textbf{H}} & \leq \varphi^{ t, \textbf{H}}(x), \label{eq:Th1_1}\\
  \varphi^{ t, \textbf{H}}(z)  &  \leq 0, \label{eq:Th1_2}
\end{align}
for all $x\in  \cZ_t(\textbf{H}),  \ z\in \cC$.
From the Riesz representation theorem, there exists $h^{ t, \textbf{H}} \in \bR^{N\times M}$ such that $\varphi^{t, \textbf{H}} (x) = \langle h^{ t, \textbf{H}}, x \rangle$ for all $x \in \bR^{N \times M}$, where $\langle x,y\rangle:=\sum_{j=1}^N\sum_{i=1}^M x_{i j}y_{i j}$ for all $x \in \bR^{N \times M}, y \in \bR^{N \times M}$ denotes the Frobenius inner product in $\bR^{N \times M}$.
From \eqref{eq:Th1_2}, we have that $\langle h^{ t, \textbf{H}}, z \rangle \leq 0$ for all $z \in \cC$, and therefore, $h^{t, \textbf{H}}_{ij} \geq 0$ for $i =1, \dots, M$ and $j =1, \dots, N$.

Since, in view of  \eqref{eq:Th1_1} we have that $h^{t, \textbf{H}}\neq 0$, we may assume without loss of generality that $\sum_{i=1}^M h^{ t, \textbf{H}}_{i j} =1$.

Also in view of \eqref{eq:Th1_1}, we deduce that
\begin{align*}
 0&< \epsilon_{t, \textbf{H}}  \leq  \sum_{j=1}^N \sum_{i=1}^M\, h^{ t, \textbf{H}}_{i j}\bE^{\bP}_t\Big[\eta\sum_{s=t+1}^T  H_s^i\Big](\omega_j) =\sum_{j=1}^N\bE^{\bP}_t\Big[\eta\sum_{s=t+1}^T \widetilde{H}_s(j)\Big](\omega_j)
\end{align*}
for all $\eta \in \cE^{\gamma}_t$, where $\widetilde{H}(j):=\sum_{i=1}^Mh^{t, \textbf{H}}_{ij}H^i$ for $j =1, \dots, N$.
Therefore, there exists  $j\in\{1, \dots, N\}$ and an $\epsilon>0$ so that
\begin{equation*}
  0< \epsilon<\bE^{\bP}_t\Big[\eta\sum_{s=t+1}^T \widetilde{H}_s(j)\Big](\omega_j) \, .
\end{equation*}
Let us define
\begin{equation*}
\epsilon':=\inf_{\eta\in\cE^\gamma_t}\frac{\epsilon}{\bE^{\bP}_t[\eta](\omega_j)}.
\end{equation*}
Since any $\eta\in \cE^{\gamma}_t$ is strictly positive and $\sup_{\eta \in \cE^{\gamma}_t}\bE^{\bP}_t[\eta](\omega_j)<\infty$, it follows that
 \begin{align*}
   0< \epsilon' &\leq \frac{\bE^{\bP}_t\Big[\eta\sum_{s=t+1}^T \widetilde{H}_s(j)\Big](\omega_j)}{\bE^{\bP}_t[\eta](\omega_j)   } =\bE^{\bQ}_t\Big[\sum_{s=t+1}^T \widetilde{H}_s(j)\Big](\omega_j)
 \end{align*}
 for all $\bQ \in \cQ^{\gamma}_t$.
 Consequently, taking infimum with respect to $\bQ \in \cQ^{\gamma}_t$ and applying Theorem~\ref{theorem:DCRMrep},  we get
\begin{equation}\label{eq:tmii}
0<  \epsilon'  \leq -\rho^{\gamma}_t(\widetilde{H}(j))(\omega_j) \, .
\end{equation}
The set $\cH(t)$ is a convex cone, hence $\widetilde{H}(j) \in \cH(t)$.
Thus, in view of \eqref{eq:tmii}, the cash flow $\widetilde{H}(j)\in \cH(t)$ violates the NGD condition for $\cH(t)$ at time $t$ and level $\gamma$, which is a contradiction.

Hence, $\cZ_t(\textbf{H}) \cap \cC\neq \emptyset$ for all $t \in \{0, \dots, T-1\}$ and $\textbf{H}\in \cH(t)\times \cdots \times \cH(t)$.
Consequently, for each $t\in \cT$,  $\textbf{H}\in \cH(t)\times \cdots \times \cH(t)$, the set
\begin{align*}
F_t(\textbf{H} ):&=  \left\{\eta \in \cE^{\gamma}_t : \; \bE^{\bP}_t \Big[\eta\sum_{s=t+1}^T H^i_s \Big](\omega_j)\leq 0, \; i =1,2, \dots, M,\;j=1, 2, \dots, N\right\}
\end{align*}
is nonempty.

Let us define the following mapping
\begin{equation*}
 \Psi_{t, \textbf{H}}(\zeta):= \Bigg[\bE^{\bP}_t \Big[\zeta\sum_{s=t+1}^T H^i_s \Big](\omega_j) \Bigg]_{j=1, \dots, N; i=1, \dots, M},
 \end{equation*}
 for any random variable $\zeta\, : \, \Omega \to \bR.$

Since,
\begin{align*}
\cZ_t(\textbf{H})\cap \cC &=  \left\{ \Bigg[\bE^{\bP}_t \Big[\eta\sum_{s=t+1}^T H^i_s \Big](\omega_j) \Bigg]_{j=1, \dots, N ; i=1, \dots, M} \; \right. \\
& \left. : \; \eta \in \cE^{\gamma}_t, \; \bE^{\bP}_t\Big[\eta\sum_{s=t+1}^T H^i_s \Big](\omega_j) \leq 0, \; i=1,2,\dots,M, \; j =1,2, \dots, N\right\},
\end{align*}
we have that $\Psi_{t, \textbf{H}}^{-1}(\cZ_t(\textbf{H})\cap\cC)=F_t( \textbf{H} )$.
Recall that $\cZ_t(\textbf{H})$ is compact and hence $\cZ_t(\textbf{H})\cap \cC$ is closed, and since  $\Psi_{t, \textbf{H}}$ is continuous, we conclude that
$F_t(\textbf{H} )$ is closed.

Now, note that
\begin{equation*}
F_t( \textbf{H} ) = \bigcap_{i =1}^M \Big\{\eta \in \cE^{\gamma}_t: \bE^{\bP}_t \Big[ \eta\sum_{s=t+1}^T H^i_s\Big](\omega_j) \leq 0 \; \textrm{for all} \; j=1, 2, \dots, N\Big\}\neq \emptyset.
\end{equation*}
Therefore, the family of subsets
\begin{equation*}
\Big\{\eta \in \cE^{\gamma}_t: \bE^{\bP}_t \Big[ \eta\sum_{s=t+1}^T H_s\Big](\omega_j) \leq 0  \; \textrm{for all} \; j=1, 2, \dots, N\Big\}_{H \in \cH(t)} \subseteq \cE^\gamma_t
 \end{equation*}
 satisfies the finite intersection property\footnote{The family of sets   $\{Y_{i}\}_{i \in \cI}$ has finite intersection property if $\bigcap_{i \in \cI'}Y_{i}$ is non-empty for any finite $\cI'\subset\cI$.}.
Since  $\cE^{\gamma}_t$ is compact, we have by Lemma~\ref{thm:FiniteIntersection} that the set
\begin{equation}\label{eq:Th1_3}
 U_{ t}:=\bigcap_{H \in \cH(t)}\Big\{\eta \in \cE^{\gamma}_t: \bE^{\bP}_t \Big[ \eta\sum_{s=t+1}^T H_s\Big]\leq 0 \Big\}
\end{equation}
is nonempty.
Hence, there exists an $\hat \eta \in \cE^{\gamma}_t$ so that $\bE^{\bP}_t [\hat \eta\sum_{s=t+1}^T H_s](\omega) \leq 0$ for all $\omega \in \Omega$ and $H \in \cH(t)$.
Now, let $\hat \bQ$ be a measure corresponding to $\hat \eta$, so that $\hat \bQ\in \cQ^\gamma_t.$
Using the abstract version of Bayes rule applied to $\hat \bQ$  we get
\begin{equation*}
  \bE^{\hat \bQ}_t \Big[ \sum_{s=t+1}^T H_s\Big]=\frac{\bE^{\bP}_t [\hat \eta\sum_{s=t+1}^T H_s]}{\bE^{\bP}_t[\hat \eta]}\leq 0
\end{equation*}
for all  $H \in \cH(t)$. So, in view of Definition~\ref{Def:RiskNeutralMeasure0} and Lemma~\ref{lemma: RHandRH0}, we see that $\hat \bQ\in \cR(\cH(t)).$
Thus,   $\cR(\cH(t)) \cap \cQ^{\gamma}_t \neq \emptyset$. The theorem is proved.
\end{proof}

Since $\cR(\cH(t)) \cap \cQ^{\gamma}_t\neq\emptyset$ implies $\cR(\cH(t))\neq \emptyset$, it is immediate from Proposition~\ref{theorem:FFTAP} and Theorem~\ref{theorem:NGDFTAP} that if NGD  holds, then the no-arbitrage condition also holds true.

\section{Dynamic ask and bid prices via DCAI}\label{sec:DynamicBidAskviaDCAI}
In this section, we derive the dynamic bid and ask prices for a derivative contract via DCAIs.
We start by constructing the set of extended cash flows that will be used to derive the good-deal ask and bid prices.
Let $D\in L^0$ be  a cash flow associated to a \emph{derivative contract}.
For a fixed $t \in \{0, \dots, T-1\}$, $D \in  L^0$, and an $\cF_t$-measurable random variable $X_t$, we define the following sets
\begin{align}
  \widehat{\cH}(t): =&\Big\{ \Big(0, \dots,0,\xi_t X^*_t, H_{t+1}-\xi_t D^*_{t+1}, \dots, H_T-\xi_t D^*_T\Big)\nonumber \\
& \qquad \qquad \qquad : H \in \cH(t),\; \xi_t \;\text{is}\; \cF_t\text{-measurable}, \;\xi_t \geq 0\Big\},\label{Eq:DefExtendedMarket} \\
  \overline{\cH}(t): =&\Big\{ \Big(0, \dots,0,-\xi_t X^*_t,H_{t+1}+ \xi_t D^*_{t+1}, \dots,H_T+\xi_t D^*_T\Big) \nonumber\\
& \qquad \qquad \qquad : H \in \cH(t), \;
\xi_t \;\text{is}\; \cF_t\text{-measurable}, \;\xi_t \geq 0\Big\}\label{Eq:DefExtendedMarket2},
\end{align}
where $X^*_t:=B^{-1}_tX_t$.
The pair $(\widehat{\cH}(t), \overline{\cH}(t))$ is interpreted as the \emph{set of extended cash flows}.

In particular, a cash flow  in $\widehat{\cH}(t)$ equals to the sum of a position in the underlying market $\cH(t)$ and a nonnegative \emph{static} position of $\xi_t$ units in the discounted cash flow \\ $(0, \dots, 0, X^*_t, -D^*_{t+1}, \dots, -D^*_{T})$.
Similarly, a cash flow  in $\overline{\cH}(t)$ equals to the sum of a position in the underlying market\footnote{Recall that $\cH(t)$ denotes the set of hedging cash flows initiated at time $t$.} $\cH(t)$ and a nonnegative \emph{static} position of
$\xi_t$ units in the discounted cash flow $(0, \dots, 0, -X^*_t, D^*_{t+1}, \dots, D^*_{T})$.

Notice that $ \cH(t) \subset \widehat{\cH}(t)\cap \overline{\cH}(t)$.
Indeed, taking any $H\in\cH(t)$ and $\xi_t=0$ in \eqref{Eq:DefExtendedMarket} and \eqref{Eq:DefExtendedMarket2}, we get that $H\in \widehat{\cH}(t)$ and $H\in\overline{\cH}(t)$.

Similarly to Definition~\ref{Def:RiskNeutralMeasure0}, we say that a probability measure $\bQ$ is \emph{risk-neutral for} $\widehat{\cH}(t)$, respectively $\overline{\cH}(t)$,  if $\bQ \sim \bP$, and $\bE^{\bQ}_t\big[\sum_{s=t}^TH_s\big]\leq 0$ for all $H \in \widehat{\cH}(t)$, respectively for all $H\in\overline{\cH}(t)$.
Also, we say that the \emph{no-good-deal condition holds for} $\widehat{\cH}(t)$, respectively $\overline{\cH}(t)$, at time $t\in \cT$ and level $\gamma>0$, if $\rho^{\gamma}_t(H)\geq0$ for all $H \in \widehat{\cH}(t)$, respectively  $H\in\overline{\cH}(t)$.
We denote by $\cR(\widehat{\cH}(t))$, respectively, $\cR(\overline{\cH}(t))$, the set of all risk-neutral measures for $\widehat{\cH}(t)$, respectively $\overline{\cH}(t)$.

\begin{remark}\label{remark:NGDextended}
Note that $\widehat{\cH}(t)$ and $\overline{\cH}(t)$ are convex cones. Thus, we may replace $\cH(t)$ with $\widehat{\cH}(t)$ or $\overline{\cH}(t)$ in Theorem~\ref{theorem:NGDFTAP} to prove that NGD holds for $\widehat{\cH}(t)$, respectively $\overline{\cH}(t)$, at time $t \in \cT$ and level $\gamma >0$  if and only if $\cR(\widehat{\cH}(t)) \cap \cQ^{\gamma}_t =\emptyset$, respectively $\cR(\overline{\cH}(t)) \cap \cQ^{\gamma}_t =\emptyset$.

\end{remark}

For the sake of brevity, we define the mappings $\delta^+_t, \delta_t: L^0 \rightarrow L^0$ as follows
\[
  \begin{array}{rccccccccc}
  \delta^{+}_t(D) :=\big(0, & \dots, & 0, & 0, & D_{t+1}, & \dots & D_T\big), &\quad &t \in \{0, \dots, T-1\},\\
  \delta_t(D)     :=\big(0, & \dots, & 0, & D_t, & 0, & \dots, & 0 \big), &\quad &t \in \cT.
  \end{array}
\]

Next we introduce the main objects of this study -- \emph{the good-deal ask and bid prices corresponding to a given DCAI $\alpha$}:
\begin{definition}\label{def:dynamicBidAsk}
The discounted \emph{good-deal ask and bid prices}  of a derivative contract $D \in  L^0$, at level $\gamma>0$, at time $t\in \{1, \dots, T-1\}$ are defined as
\begin{align*}
\Pi^{ask, \gamma}_t(D)(\omega):&= \inf\{v\in \bR: \;\text{there exists}\; H \in \cH(t)\; \text{s.t.} \;\alpha_t(\delta_t(\textbf{1}v)+H-\delta^{+}_t(D^*))(\omega) \geq \gamma\},\\
\Pi^{bid, \gamma}_t(D)(\omega):&= \sup\{v\in \bR: \;\text{there exists}\; H \in \cH(t)\; \text{s.t.}\; \alpha_t(\delta^{+}_t(D^*)+H-\delta_t(\textbf{1}v))(\omega) \geq \gamma\},
\end{align*}
for all $\omega \in \Omega$.
\end{definition}

\begin{remark}
We stress that the good-deal prices depend on the choice of DCAI $\alpha$, level $\gamma$, and the set of hedging cash flows $\cH(t)$.
 First, we see  that, from the monotonicity property of DCAIs (D3), the good-deal ask (bid) price is non-decreasing (non-increasing) in $\gamma$.
Secondly, the good-deal ask (bid) price is non-increasing (non-decreasing) in  $\cH(t)$. This is because, as is easily seen, $\Pi^{ask, \gamma}_t(D)$ and $\Pi^{bid, \gamma}_t(D)$ satisfy
\begin{align*}
\Pi^{ask, \gamma}_t(D)(\omega)&=\inf \bigcup_{H \in \cH(t)}\{v\in \bR:  \;\alpha_t(\delta_t(\textbf{1}v)+H-\delta^{+}_t(D^*))(\omega) \geq \gamma\},\\
\Pi^{bid, \gamma}_t(D)(\omega)&= \sup \bigcup_{H \in \cH(t)}\{v\in \bR:  \;\alpha_t(\delta^{+}_t(D^*)+H-\delta_t(\textbf{1}v))(\omega) \geq \gamma\}
\end{align*}
for all $\omega \in \Omega$.
  \end{remark}

\begin{remark}
A natural question is: how should $\gamma$ be chosen to find the good-deal prices of an derivative contract?
As in Cherny and Madan \cite{MadanCherny2010} and Madan and Schoutens \cite{MadanCoconut2011,MadanSchoutensEq2011}, for a given $\alpha$, the level $\gamma$ can be calibrated from quoted prices of similar contracts. 
\end{remark}

\begin{remark}
The discounted good-deal ask price $\Pi^{ask, \gamma}_t(D)$ can be interpreted as the minimum amount of cash $v$ such that $v$ plus the  resulting hedging error is acceptable (in the sense of acceptability index $\alpha$) at least at level $\gamma$.
Similarly, the discounted good-deal bid price $\Pi^{bid, \gamma}_t(D)$ can be viewed as the maximum amount of cash $v$ such that $-v$ plus the  resulting hedging error is $\alpha$-acceptable at least at level $\gamma$.
\end{remark}

\begin{remark}
By Theorem~\ref{theorem:DCAIrep}, we have that
\begin{equation*}
\alpha_t(\delta_t(\textbf{1}v)+H-\delta^{+}_t(D^*))(\omega)=\sup\Big\{\gamma \in(0,+\infty):v+ \inf_{\mathbb{Q}\in\cQ^{\gamma}_t}\bE^{\bQ}_t\big[\sum_{s=t+1}^T H_s-D^*_s\big](\omega)\geq 0\Big\}
\end{equation*}
for all $\omega\in \Omega$, $t\in\{1, \dots, T-1\}$, and $D\in L^0$.
Since the cash flows $D^*$ and $H \in \cH(t)$ are discounted, the prices $\Pi^{ask, \gamma}(D)$ and $\Pi^{bid, \gamma}(D)$ are also discounted.
We took the liberty to denote them by $\Pi^{ask, \gamma}(D)$ and $\Pi^{bid, \gamma}(D)$ rather than $\Pi^{ask, \gamma, *}(D)$ and $\Pi^{bid, \gamma, *}(D)$ (which would agree with earlier notation) to ease exposition.
\end{remark}

\begin{proposition}
For any fixed $t \in \{1, \dots, T-1\}$, $D \in  L^0$, and $\gamma>0$, the sets
\begin{align*}
   &\{v\in \bR: \;\text{there exists}\; H \in \cH(t)\; \text{s.t.} \;\alpha_t(\delta_t(\textbf{1}v)+H-\delta^{+}_t(D^*))(\omega) \geq \gamma\},\\
   & \{v\in \bR: \;\text{there exists}\; H \in \cH(t)\; \text{s.t.}\; \alpha_t(\delta^{+}_t(D^*)+H-\delta_t(\textbf{1}v))(\omega) \geq \gamma\}
\end{align*}
are nonempty for all $\omega \in \Omega$.
\end{proposition}
\begin{proof} The proof will be done by contradiction. Towards this end let us fix $t \in \{1, \dots, T-1\}$, $D \in  L^0$, and $\gamma>0$.

Suppose that
\begin{equation*}
\alpha_t(\delta_t(\textbf{1}v)+H-\delta^{+}_t(D^*))< \gamma
\end{equation*}
for all $v \in \bR$ and $H \in \cH(t)$.
By Theorem~\ref{theorem:DCAIrep}, we have that
\begin{equation*}
\alpha_t(\delta_t(\textbf{1}v)+H-\delta^{+}_t(D^*))(\omega)=\sup\Big\{\beta\in(0,+\infty): v+\inf_{\bQ\in\cQ^{\beta}_t}\bE^{\bQ}_t\Big[\sum_{s=t+1}^T H_s-D_s\Big](\omega)\geq 0\Big\} <\gamma
\end{equation*}
for all $v \in \bR$ and $H \in \cH(t)$.
Since $\alpha$ is normalized, there exists $D' \in  L^0$ such that $\alpha_t(D')=+\infty$.
Let us define $v^*$ as
\begin{equation*}
 v^*:=\sup_{ \omega \in \Omega} \sup_{H \in \cH(t)}\Bigg\{\sup_{\bQ\in\cQ^{\gamma}_t}\bE^{\bQ}_t\Big[\sum_{s=t+1}^T D'_s\Big](\omega)-\inf_{\bQ\in\cQ^{\gamma}_t}\bE^{\bQ}_t\Big[\sum_{s=t+1}^T H_s-D_s\Big](\omega)\Bigg\}.
\end{equation*}

Note that,
$ v^*  \geq \sup_{H \in \cH(t)}\Bigg\{\bE^{\bQ}_t\Big[\sum_{s=t+1}^T D'_s-H_s+D_s\Big](\omega)\Bigg\},$
$\omega \in \Omega, \ \bQ \in \cQ^{\gamma}_t,$ and since $H=(0,0,\ldots,0)\in \cH(t)$, we have that
$$
v^* \geq \bE^{\bQ}_t\Big[\sum_{s=t+1}^TD'_s+D_s\Big](\omega) >-\infty, \qquad \omega \in \Omega, \; \bQ \in \cQ^{\gamma}_t.
$$
Next, we see that
\begin{equation*}
 v^*+\bE^{\bQ}_t\Big[\sum_{s=t+1}^TH_s-D_s\Big](\omega)\geq \bE^{\bQ}_t\Big[\sum_{s=t+1}^T D'_s\Big](\omega),
\end{equation*}
for all $\bQ \in \cQ^{\gamma}_t$, $\omega \in \Omega$, and $H \in \cH(t)$.
From the monotonicity property of $\alpha$, we have that
\begin{equation*}
\alpha_t(\delta_t(\textbf{1}v^*)+H-\delta^{+}_t(D^*)) \geq \alpha_t(D')=+\infty,
\end{equation*}
which contradicts $\alpha_t(\delta_t(\textbf{1}v)+H-\delta^{+}_t(D^*))(\omega) < \gamma$ for all $v \in \bR$.
\end{proof}

We close this section with a technical result, which provides a ``symmetry'' between ask and bid prices, that will be used later.
  \begin{lemma}\label{lemma:RepBidAskNegation}
For any $D\in L^0$, $\gamma >0$, and $t \in \{0, \dots, T-1\}$ we have that $$\Pi^{ask, \gamma}_t(D)=-\Pi^{bid, \gamma}_t(-D).$$
  \end{lemma}

\begin{proof}
Using the definitions of $\Pi^{ask, \gamma}_t(D)$ and $\Pi^{bid, \gamma}_t(D)$, we have
  \begin{align*}
    \Pi^{ask, \gamma}_t(D) &= \inf\{v\in \bR: \;\text{there exists}\; H \in \cH(t)\; \text{s.t.} \;\alpha_t(\delta_t(\textbf{1}v)+H-\delta^{+}_t(D^*)) \geq \gamma\} \\
     & = - \sup\{-v\in \bR: \;\text{there exists}\; H \in \cH(t)\; \text{s.t.} \;\alpha_t(\delta_t(\textbf{1}v)+H-\delta^{+}_t(D^*)) \geq \gamma\} \\
    & = -\sup\{v\in \bR: \;\text{there exists}\; H \in \cH(t)\; \text{s.t.} \;\alpha_t(-\delta_t(\textbf{1}v)+H-\delta^{+}_t(D^*)) \geq \gamma\}\\
    & = -\Pi^{bid, \gamma}_t(-D).
  \end{align*}
\end{proof}

\subsection{Dual representation of good-deal ask and bid prices}\label{sec:PricesRep}
We are now in position to prove a representation theorem for the discounted good-deal ask and bid prices.
Let us first make the following standing technical assumption

\smallskip
\noindent {\bf Assumption (C):}\label{Assumption:ContinuityRho}
\textup{The mapping $\gamma \mapsto \rho^{\gamma}$ is continuous.}

\smallskip \noindent
In Section~\ref{sec:PricingDGLR}, we prove that the dynamic Gain-Loss Ratio satisfies this assumption.

We proceed by showing that, for any derivative contract $D\in L^0$, the prices $\Pi^{ask}_t(D)$ and $\Pi^{bid}_t(D)$ have useful representations in terms of the sets $\cR(\cH(t))$ and $\cQ^{\gamma}_t(\cH(t))$.

\begin{theorem}\label{Thm:RepSpotAskBid}
The discounted good-deal ask and bid prices of a derivative contract $D \in  L^0$, at level $\gamma>0$, at time $t\in \{1, \dots, T-1\}$ satisfy
\begin{align*}
 \Pi^{ask, \gamma}_t(D) &=\sup_{ \mathbb{Q} \in \cQ^{\gamma}_t \cap \cR(\cH(t))} \bE^{\mathbb{Q}}_t \Big [ \sum _{s=t+1}^T D^*_s \Big] \, ,\\
 \Pi^{bid, \gamma}_t(D) &=\inf_{ \mathbb{Q} \in \cQ^{\gamma}_t \cap \cR(\cH(t))} \bE^{\mathbb{Q}}_t \Big [ \sum _{s=t+1}^T D^*_s \Big] \, .
\end{align*}
\end{theorem}

\begin{proof}
In view of Lemma~\ref{lemma:RepBidAskNegation}, it is enough to prove that the theorem holds for $\Pi^{ask, \gamma}(D)$.
Let $D \in  L^0$, $\gamma >0$, and $t \in \{1, \dots, T-1\}$.
We first show that
\begin{equation*}
\Pi^{ask, \gamma}_t(D) \leq \sup_{ \mathbb{Q} \in \cQ^{\gamma}_t \cap \cR(\cH(t))} \bE^{\mathbb{Q}}_t \Big [ \sum _{s=t+1}^T D^*_s \Big] \, .
\end{equation*}

Using Theorem~\ref{theorem:Duality} and Lemma~\ref{lemma:ContinuityRhoSet}, as well as the continuity and monotonicity of the map $\gamma \mapsto \rho^{\gamma}$, we obtain
\begin{equation}\label{eq:ThemRep7}
   \Pi^{ask, \gamma}_t(D)(\omega) =\inf\Big\{ v \in \bR: \; \text{there exists} \;H \in \cH(t)\; \text{s.t.} \; \rho^{\gamma} _t(\delta_t(\textbf{1}v)+H-\delta^{+}_t(D))(\omega)  \leq 0 \Big\}
\end{equation}
for all $\omega \in \Omega$.

Now fix an $\cF_t$-measurable random variable $X_t$, and let $\cP^t:=\{P^t_1, P^t_2, \dots, P^t_{n_t}\}$ be the unique partition that generates $\cF_t$.
Fix $P^t_i\neq \emptyset$ and let $\omega_i \in P^t_i$.
Then $\1_{P^t_i}(\omega)X_{t}(\omega_i)=\1_{P^t_i}(\omega)X_{t}(\omega)$ for all $\omega \in \Omega$.
Using \eqref{eq:ThemRep7}, we have that $\Pi^{ask, \gamma}_t(D)(\omega_i) >X^*_{t}(\omega_i)$ if and only if
\begin{equation*}
X^*_{t}(\omega_i) \notin \Big\{ v \in \bR:\;\text{there exists} \; H \in \cH(t)\; \text{s.t.} \; \rho^{\gamma} _t(\delta_t(\textbf{\textbf{1}}v)+H-\delta^{+}_t(D^*))(\omega_i)  \leq0 \Big\}.
\end{equation*}
Equivalently,
\begin{equation*}
  \1_{P^t_i}(\omega)\rho^{\gamma}_t(\delta_t(\textbf{\textbf{1}}X^*_{t}(\omega_i))+H-\delta^{+}_t(D^*))(\omega) > 0, \qquad H \in \cH(t), \omega \in P^t_i.
\end{equation*}
By property (A2) in Definition~\ref{Def:DCRM} of $\rho^{\gamma}$, it follows that
\begin{equation*}
  \rho^{\gamma}_t(\delta_t(\textbf{1}X^*_{t})+H-\delta^{+}_t(D^*))(\omega) \geq 0,\quad \;H \in \cH(t), \omega \in P^t_i.
\end{equation*}

By Theorem~\ref{theorem:DCRMrep} and property (A6),  we deduce that
\begin{equation*}
\rho^{\gamma}_t(\delta_t(\textbf{1}\xi_tX^*_{t})+\xi_t H-\xi_t \delta^{+}_t(D^*)) \geq  0
\end{equation*}
for any $H \in \cH(t)$ and any nonnegative $\cF_t$-measurable random variable $\xi_t$.
Since $\cH(t)$ is closed under multiplication of nonnegative $\cF_t$-measurable random variables, the inequality above is equivalent to
\begin{equation*}
\rho^{\gamma}_t(\xi_t\delta_t(\textbf{1} X^*_{t})+ H-\xi_t \delta^{+}_t(D^*))\geq 0
\end{equation*}
for any $H \in \cH(t)$ and any nonnegative $\cF_t$-measurable random variable $\xi_t$.

Therefore, by the definition of $\widehat{\cH}(t)$, we see that
\begin{equation*}
\rho^{\gamma}_t(\widehat{H})\geq  0,   \quad  \;\widehat{H} \in \widehat{\cH}(t),
\end{equation*}
and hence NGD holds  for $\widehat{\cH}(t)$, at time $t$ and level $\gamma$.
It follows that $\cR(\widehat{\cH}(t)) \cap \cQ^{\gamma}_t \neq \emptyset$ (see Remark~\ref{remark:NGDextended}).
Let $\bQ^* \in \cR(\widehat{\cH}(t)) \cap \cQ^{\gamma}_t$.

From the definition of $\cR(\widehat{\cH}(t))$, we have that
\begin{equation}\label{eq:RepThm3}
\bE^{\bQ^*}_t\Big[\sum_{u=t+1}^T(H_u -\xi_t D^*_u)\Big]+\xi_t X^*_{t}\leq 0
\end{equation}
 for all $H \in \cH(t)$ and all nonnegative $\cF_t$-measurable random variables $\xi_t$.
Note that $\cR(\cH(t)) \supseteq \cR(\widehat{\cH}(t))$ since $\cH(t)\subset \widehat{\cH}(t)$.
Thus, $\bQ^* \in \cQ^{\gamma}_t\cap \cR(\cH(t))$.
Because $0 \in \cH(t)$, we may let $H=0$ in \eqref{eq:RepThm3} to conclude that, if $\Pi^{ask, \gamma}_t(D) >X^*_{t}$, then there exists $\bQ^* \in  \cQ^{\gamma}_t\cap\cR(\cH(t))$ such that

\begin{equation*}
\bE^{\bQ^*}_t\Big[\sum_{s=t+1}^TD^*_s\Big]\geq X^*_{t}.
\end{equation*}
Since $X_t$ is arbitrary,
\begin{equation}
\Pi^{ask, \gamma}_t(D) \leq \sup_{ \mathbb{Q} \in \cQ^{\gamma}_t \cap \cR(\cH(t))} \bE^{\mathbb{Q}}_t \Big [ \sum _{s=t+1}^T D^*_s \Big] \, .
\label{eq:askLeqsup}
\end{equation}

We proceed by showing that
\begin{equation*}
\Pi^{ask, \gamma}_t(D)\geq \sup_{ \mathbb{Q} \in \cQ^{\gamma}_t \cap \cR(\cH(t))} \bE^{\mathbb{Q}}_t \Big [ \sum _{s=t+1}^T D^*_s \Big] \, .
\end{equation*}
 By Theorem~\ref{theorem:DCRMrep},
 \begin{equation}
   \rho^{\gamma}_t(H-\delta^{+}_t(D^*))=\underset{ \bQ \in \cQ^{\gamma}_t}{\sup} \bE_t^{\bQ}\Big[\sum_{s=t+1}^TD^*_s-H_s\Big] \, .\label{eq:Thm2rho1}
 \end{equation}
Also, we have that
 \begin{align*}
  \underset{ \bQ \in \cQ^{\gamma}_t}{\sup} \bE_t^{\bQ}\Big[\sum_{s=t+1}^T D^*_s-H_s\Big]
  &\geq \underset{ \bQ \in \cQ^{\gamma}_t\cap \cR(\cH(t))}{\sup} \bE_t^{\bQ}\Big[\sum_{s=t+1}^T D^*_s-H_s\Big]
  \geq \bE_t^{\bQ}\Big[\sum_{s=t+1}^TD^*_s\Big] \, , \label{eq:Thm2rho2}
 \end{align*}
 for all $H \in \cH(t)$ and $\bQ \in \cQ^{\gamma}_t\cap\cR(\cH(t))$.
Therefore,
  \begin{equation}\label{Eq:InequalityRepThm1}
    \rho^{\gamma}_t(H-\delta^{+}_t(D^*)) \geq \bE_t^{\bQ}\Big[\sum_{s=t+1}^TD^*_s\Big] \, , \quad \cH \in \cH(t), \; \bQ \in \cQ^{\gamma}_t \cap \cR(\cH(t)) \, .
  \end{equation}
Note that
\begin{equation*}
\Pi^{ask, \gamma}_t(D)(\omega)= \underset{H \in \cH(t)}\inf\inf\{v\in \bR: \;\alpha_t(\delta_t(\textbf{1}v)+H-\delta^{+}_t(D^*))(\omega) \geq \gamma\}
\end{equation*}
for all $\omega \in \Omega$.
In view of  Theorem~\ref{theorem:Duality},
\begin{equation*}
\Pi^{ask, \gamma}_t(D)= \underset{H \in \cH(t)}\inf\rho^{\gamma}_t(H-\delta^{+}_t(D^*)) \, .
\end{equation*}
Hence, applying \eqref{Eq:InequalityRepThm1} we see that
\begin{equation}
\Pi^{ask, \gamma}_t(D)\geq \sup_{ \mathbb{Q} \in \cQ^{\gamma}_t \cap \cR(\cH(t))} \bE^{\mathbb{Q}}_t \Big [ \sum _{s=t+1}^T D^*_s \Big] \, . \label{eq:askGeqsup}
\end{equation}
In virtue of \eqref{eq:askLeqsup} and \eqref{eq:askGeqsup}, we conclude the  proof.
\end{proof}

Let us now make a few remarks regarding Theorem~\ref{Thm:RepSpotAskBid}.
\begin{remark}
  If NGD holds false for $\cH(t)$, at time $t \in \{1, \dots, T-1\}$, at level $\gamma$, then
  \begin{align*}
 \Pi^{ask, \gamma}_t(D)(\omega) &=-\infty ,\\
 \Pi^{bid, \gamma}_t(D)(\omega) &=\infty,
\end{align*}
for all $\omega \in \Omega$ and $D \in L^0$.
\end{remark}

\noindent
In the next remark, we treat the case in which the markets are frictionless and complete.

\begin{remark}\label{remark:complete}
  If, for $t \in\{1, \dots, T-1\}$, the set of hedging cash flows $\cH(t)$ satisfies the no-arbitrage condition, and $\cH(T-1)$ is complete (for any $D \in L^0$, there exists $H \in \cH(T-1)$ so that $H_T=D_T$), then it follows from the Fundamental Theorems of Asset Pricing  that $\cR(\cH(t)) \neq \emptyset$, for $t =1, 2, \dots, T-2$, and $\cR(\cH(T-1))=\{\bQ^*\}$.
  Since $\cR(\cH(0)) \subseteq \cdots \subseteq \cR(\cH(T-1))$, we have that $\cR(\cH(t))=\{\bQ^*\} \neq \emptyset$ for $t=0, 1, \dots, T-2$.
  By Theorems~\ref{theorem:NGDFTAP} and~\ref{Thm:RepSpotAskBid}, if NGD holds then the good-deal ask and bid prices of a derivative contract $D \in L^0$, at time $t \in \cT$ and level $\gamma>$, satisfy
  \begin{align*}
    \Pi^{ask, \gamma}_t(D)=\Pi^{bid, \gamma}_t(D)=\bE_t^{\bQ^*}\Big[\sum_{s=t+1}^TD^*_s\Big] \, .
  \end{align*}
  Notice that, naturally, the good-deal prices no longer depend on the acceptance level $\gamma$.
\end{remark}

\begin{remark}
Let us consider the sets of extended cash flows associated with good-deal prices $\Pi^{ask, \gamma}_t(D)$ and $\Pi^{bid, \gamma}_t(D)$:
\begin{align*}
  \widehat{\cH}(t) =&\Big\{ \Big(0, \dots,0,\xi_t \Pi^{ask, \gamma}_t(D), H_{t+1}-\xi_t D^*_{t+1}, \dots, H_T-\xi_t D^*_T\Big)\nonumber \\
& \qquad : H \in \cH(t),\; \xi_t \;\text{is}\; \cF_t\text{-measurable}, \;\xi_t \geq 0\Big\},\\
  \overline{\cH}(t) =&\Big\{ \Big(0, \dots,0,-\xi_t \Pi^{bid, \gamma}_t(D),H_{t+1}+ \xi_t D^*_{t+1}, \dots,H_T+\xi_t D^*_T\Big) \nonumber\\
& \qquad : H \in \cH(t), \;
\xi_t \;\text{is}\; \cF_t\text{-measurable}, \;\xi_t \geq 0\Big\}.
\end{align*}
If $\cH(t)$ is frictionless and complete (and therefore linear), and NGD holds, then as in Remark~\ref{remark:complete}, we have that $\Pi_t(D):=\Pi^{ask, \gamma}_t(D)= \Pi^{bid, \gamma}_t(D)$.
In this case, the set
\begin{align*}
\widehat{\cH}(t)+\overline{\cH}(t)&=\Big\{ \Big(0, \dots,0,\xi_t \Pi_t(D), H_{t+1}-\xi_t D^*_{t+1}, \dots, H_T-\xi_t D^*_T\Big)\nonumber \\
& \qquad : H \in \cH(t),\; \xi_t \;\text{is}\; \cF_t\text{-measurable}\Big\}
\end{align*}
 is a linear space.
 Whenever $\Pi^{ask, \gamma}_t(D)> \Pi^{bid, \gamma}_t(D)$, as in our general case, we have that
 \begin{align*}
\widehat{\cH}(t)+\overline{\cH}(t)&=\Big\{ \Big(0, \dots,0,\xi_t\Pi^{ask, \gamma}_t(D)-\phi_t\Pi^{bid, \gamma}_t(D), H_{t+1}-(\xi_t-\phi_t) D^*_{t+1}, \\
& \qquad \dots, H_T-(\xi_t-\phi_t) D^*_T\Big):\ H \in \cH(t),\; \xi_t, \phi_t \;\text{is}\; \cF_t\text{-measurable}, \;\xi_t, \phi_t \geq 0\Big\}
\end{align*}
 is only a convex cone.
This is one of the main reasons why we call this approach \emph{dynamic conic finance}.
\end{remark}

\begin{remark}\label{remark:SuperhedgingPrices}
In view of Lemma~\ref{lemma:NAinterval} and Theorem~\ref{Thm:RepSpotAskBid}, if NGD is satisfied then $\Pi^{bid, \gamma}_t(D)$ and $\Pi^{ask, \gamma}_t(D)$ are within the lower and upper no-arbitrage bounds.
Specifically, we have that
\begin{equation*}
  \inf_{ \bQ \in \cR(\cH(t))} \bE^\bQ_t\Big[\sum_{s=t+1}^T D^*_t\Big] \leq \Pi^{bid, \gamma}_t(D) \leq \Pi^{ask, \gamma}_t(D) \leq \sup_{ \bQ \in \cR(\cH(t))} \bE^\bQ_t\Big[\sum_{s=t+1}^T D^*_t\Big].
\end{equation*}
\end{remark}

\subsection{Good-deal forward ask and bid prices}

In this section, we define the good-deal forward ask and bid prices, and then prove a representation theorem for them.
In this subsection we suppose that the risk-free interest rate $r$ is deterministic.

\begin{definition}\label{def:dynamicBidAsk-Forward}
The \emph{good-deal ask and bid forward prices}, with delivery at time $T$, written at time $t\in \{1, \dots, T-1\}$,  of a derivative contract $D \in  L^0$, at level $\gamma>0$ are defined as
\begin{align}
F^{ask, \gamma, T}_t(D)(\omega):&= \inf\{f\in \bR: \;\exists\; H \in \cH(t)\; \text{so that}  \nonumber \\
 & \qquad \qquad\qquad\qquad \;\alpha_t(\delta_T(\textbf{1}B^{-1}_Tf)+H-\delta^{+}_t(D^*))(\omega) \geq \gamma\},\label{eq:forward1}\\
F^{bid, \gamma, T}_t(D)(\omega):&= \sup\{f\in \bR: \;\exists\; H \in \cH(t)\; \text{so that}\; \nonumber  \\
 & \qquad \qquad\qquad\qquad \alpha_t(-\delta_T(\textbf{1}B^{-1}_T f)+H+\delta^{+}_t(D^*))(\omega) \geq \gamma\}\label{eq:forward2}
\end{align}
for all $\omega \in \Omega$.
\end{definition}
Notice that the cash flow $\delta_T(\textbf{1}B^{-1}_Tf)+H-\delta^{+}_t(D^*)$ represents an exchange of a cash payment $f$ at time $T$ for a discounted cash flow $D$ that is hedged with $H$. The good-deal forward ask price at level $\gamma$ is the minimum amount of cash $f$ at time $T$ so that $\delta_T(\textbf{1}B^{-1}_Tf)+H-\delta^{+}_t(D^*)$ is acceptable at level $\gamma$ at time $t$.

We now give the representation theorem for the good-deal forward ask and bid prices.
\begin{theorem}\label{Thm:ForwardHedging}
The good-deal ask and bid forward prices of a derivative contract $D\in L^0$, with delivery at time $T$, written at time $t\in \{1, \dots, T-1\}$ and level $\gamma>0$,  satisfy
\begin{align*}
F^{ask, \gamma, T}_t(D)= B_T\Pi^{ask, \gamma}_t(D), \qquad
F^{bid, \gamma, T}_t(D)= B_T\Pi^{bid, \gamma}_t(D).
\end{align*}
\end{theorem}

\begin{proof}
For any $f\in \bR$, denote by $f^*$ the term $B^{-1}_Tf$.
Since $B_T$ is deterministic, we may write \eqref{eq:forward1} and \eqref{eq:forward2} as
\begin{align*}
F^{ask, \gamma, T}_t(D)(\omega)&=B_T  \inf\{f^*\in \bR: \;\exists\; H \in \cH(t)\; \text{and} \;\alpha_t(\delta_T(\textbf{1}f^*)+H-\delta^{+}_t(D^*))(\omega) \geq \gamma\},\\
 F^{bid, \gamma, T}_t(D)(\omega)&=B_T \sup\{f^*\in \bR: \;\exists\; H \in \cH(t)\; \text{and}\; \alpha_t(-\delta_T(\textbf{1}f^*)+H+\delta^{+}_t(D^*))(\omega) \geq \gamma\}.
\end{align*}
Since $\alpha$ satisfies the translation  invariance property (Property (D6) in Appendix~\ref{sec:PricingDCAI}), we have that $\alpha_t(\delta_T(\textbf{1}f^*)+H-\delta^{+}_t(D^*))=\alpha_t(\delta_t(\textbf{1}f^*)+H-\delta^{+}_t(D^*))$.
Therefore,
\begin{align*}
F^{ask, \gamma, T}_t(D)(\omega)&=B_T  \inf\{f^*\in \bR: \;\exists\; H \in \cH(t)\; \text{and} \;\alpha_t(\delta_t(\textbf{1}f^*)+H-\delta^{+}_t(D^*))(\omega) \geq \gamma\},\\
 F^{bid, \gamma, T}_t(D)(\omega)&=B_T \sup\{f^*\in \bR: \;\exists\; H \in \cH(t)\; \text{and}\; \alpha_t(-\delta_t(\textbf{1}f^*)+H+\delta^{+}_t(D^*))(\omega) \geq \gamma\}.
\end{align*}
By Theorem~\ref{Thm:RepSpotAskBid} we conclude that the claim holds true.
\end{proof}

\begin{remark}
If $r$ is deterministic and the set of hedging cash flows $\cH(t)$ forms a market that is frictionless, complete, and arbitrage-free, then $\cR(\cH(t))$ is a singleton, say $\{\bQ^*\}$, and so by Theorem~\ref{Thm:ForwardHedging} we have that $F^{ask, \gamma, T}_t(D)=F^{bid, \gamma, T}_t(D)=B_T \bE^{\bQ^*}_t\big[\sum_{u=t+1}^T D^*_u \big]$.
This is compatible with the classic result that states that in a frictionless, complete, and arbitrage-free market the discounted forward price $f^T_t(D)$ of a derivative contract $D$, with delivery at time $T$, written at time $t \in \{1, \dots, T-1\}$, is given as
\begin{equation*}
  f^T_t(D)=B_TS_t(D),
\end{equation*}
where $S(D)$ is the discounted risk-neutral spot price given by $S_t(D)=\bE^{\bQ^*}_t\big[\sum_{u=t+1}^T D^*_u \big]$.
Also, from Theorem~\ref{Thm:ForwardHedging}, we see that the relationship between the good-deal ask and bid forward prices is classic, in the sense that
\begin{equation*}
 \frac{F^{ask, \gamma, T}_t(D)}{\Pi^{ask, \gamma}_t(D)}=\frac{F^{bid, \gamma, T}_t(D)}{\Pi^{bid, \gamma}_t(D)}=\frac{ f^T_t(D)}{S_t(D)}, \quad \gamma \in (0, \infty), \; D \in L^0.
\end{equation*}
\end{remark}

\section{Pricing with the dynamic Gain-Loss Ratio}\label{sec:PricingDGLR}

In this section, we first prove some auxiliary results that hold for general DCAIs.
Then, we particularize these results to the very important special case of DCAI, namely to the dynamic Gain-Loss Ratio (dGLR).
Finally, we apply the pricing and hedging results developed in earlier sections using dGLR to path-dependent options.
In this section we assume that $r=0$ without loss of generality.

\subsection{Characterization of DCAIs}
Recall that for every normalized and right-continuous DCAI $\alpha$ there exist a family ${\cQ}= (\big(\cQ^{\gamma}_t\big)_{t \in \cT})_{\gamma \in (0, \infty)}$  of dynamically consistent sequences of sets of probability measures that is increasing (in $\gamma$), such that \eqref{eq:RM-Representation} holds (see Appendix~\ref{sec:PricingDCAI}).
We say that a family $\cQ$ of dynamically consistent sequences of sets of probability measures that is increasing (in $\gamma$) \emph{corresponds} to a given normalized and right-continuous DCAI $\alpha$ if $\cQ$ satisfies \eqref{eq:RM-Representation}.

\begin{lemma}\label{lemma:CorrFamilyQ}
Suppose that $\alpha$ is a normalized and right-continuous DCAI.
A family $\cQ$ corresponds to $\alpha$ if and only if $\cQ\in \mathfrak{Q}^\alpha$, where$\, $\footnote{We will generically denote by ${\cU}=\left (\big(\cU^{\gamma}_t\big)_{t \in \cT}\right)_{\gamma \in (0, \infty)}$ a family of dynamically consistent sequences of sets of probability measures that is increasing in $\gamma$.}
    \begin{align*}
\mathfrak{Q}^{\alpha}:=&\Big\{ \cU \ : \  \alpha_t(D)(\omega)\geq\gamma \quad  \text{if and only if}\\
 & \qquad \inf_{\bQ\in\cU^{\gamma}_t}\bE^{\bQ}_t\Big[\sum_{s=t}^T D_s\Big](\omega)\geq 0, \; \omega \in \Omega, \; \gamma \in (0, \infty), \; t\in \cT, \; D\in L^0 \Big\}.
  \end{align*}
   \end{lemma}

\begin{proof}

\noindent
($\Longleftarrow$)
Let $\cU \in \mathfrak{Q}^{\alpha}$.
We fix $t \in \cT$, $D \in L^0$, and $\omega \in \Omega$.
Define the set
\begin{equation*}
  \Gamma(\cU):=\Big\{\beta \in(0,\infty): \inf_{\bQ\in\cU^{\beta}_t}\bE^{\bQ}_t\Big[\sum_{s=t}^T D_s\Big](\omega)\geq 0\Big\}.
\end{equation*}
We may assume that $\Gamma(\cU)\neq \emptyset$ and $\alpha_t(D)(\omega)<\infty$.
Otherwise, it is clear that $\cU$ satisfies \eqref{eq:RM-Representation}.

Observe that if $\gamma \in \Gamma(\cU)$, then $\alpha_t(D)(\omega) \geq \gamma$.
So $\alpha_t(D)(\omega)$ is an upper bound of $\Gamma(\cU)$.
If we let $\beta':= \alpha_t(D)(\omega)$, then $\beta' \in \Gamma(\cU)$, and so \eqref{eq:RM-Representation} is satisfied.

\noindent
($\Longrightarrow$)
Now, suppose $\cU$ satisfies \eqref{eq:RM-Representation}, and let $\gamma \in (0, \infty)$.
If
\begin{equation}\label{eq:CorrFamilyQ}
  \inf_{\bQ\in\cU^{\gamma}_t}\bE^{\bQ}_t\Big[\sum_{s=t}^T D_s\Big](\omega)\geq 0,
\end{equation}
then $\gamma \in \Gamma(\cU)$.
By \eqref{eq:RM-Representation}, we have that $\alpha_t(D)(\omega) \geq \gamma$.

Assume $\alpha_t(D)(\omega) \geq\gamma$.
If $\alpha_t(D)(\omega) >\gamma$, then \eqref{eq:CorrFamilyQ} is satisfied because $\cU^{\gamma}$ is increasing in $\gamma$.

Next, suppose that $\alpha_t(D)(\omega)=\gamma$ and $\gamma \notin \Gamma(\cU)$.
By Theorem~\ref{theorem:Duality}, the mapping
\begin{equation*}
  \gamma  \longmapsto\inf_{\bQ\in\cU^{\gamma}_t}\bE^{\bQ}_t\Big[\sum_{s=t}^T D_s\Big](\omega)
\end{equation*}
 is left-continuous and monotone decreasing.
 Thus, by left-continuity of $\rho$, there exists $\epsilon>0$ so that $\gamma- \epsilon \notin \Gamma(\cU)$.
By monotonicity and because $\cU$ satisfies \eqref{eq:RM-Representation}, we deduce that $\alpha_t(D)(\omega) \leq \gamma-\epsilon$.
This implies that $\epsilon \leq0$, which is a contradiction.
 Hence, we have that \eqref{eq:CorrFamilyQ} holds, and thus
 $\cU \in \mathfrak{Q}^{\alpha}$.
\end{proof}

\subsection{Characterization of the dGLR}\label{subsec:ChardGLR}
A performance measures  that is very popular among practitioners is the Sharpe Ratio (SR), which was introduced by Sharpe~\cite{Sharpe1964a}.
However, SR is not monotone, and hence not an acceptability index.
Moreover, as pointed out by  Bernardo~and~Ledoit~\cite{Bernardo2000} SR does not respect arbitrage, in the sense that the SR is finite even for cash-flows that exhibit arbitrage opportunities.
For this reason,~\cite{Bernardo2000} proposed the \emph{static} Gain-Loss Ratio, which is a performance measure that is unbounded for arbitrage opportunities, and, as proved in Cherny~and~Madan~\cite{ChernyMadan2009},  is also a static coherent acceptability index.
Later, Bielecki et al.~\cite{BCZ2010} extended the notion of GLR to dynamic setup, and introduced the \emph{dynamic} Gain-Loss Ratio, defined\footnote{By convention, $\textrm{dGLR}(0)=\infty$.} as follows
\begin{equation}\label{def:DLR}
\mathrm{dGLR}_t(D)(\omega):=
\begin{dcases}\frac{}{}
\frac{\bE^\bP_t[\sum_{s=t}^T D_s](\omega)}
{\bE^\bP_t[(\sum_{s=t}^T D_s )^{-}](\omega) } \, , &
\quad \mathrm{if} \quad  \bE^\bP_t\Big[ \sum_{s=t}^T D_s\Big](\omega) >0 \, ,\\
0 \, , & \quad \mathrm{otherwise} \, .
\end{dcases}
\end{equation}
It is shown in~\cite{BCZ2010} that the dGLR satisfies the conditions (D$1$)--(D$7$), and therefore it is a dynamic coherent acceptability index (see Definition~\ref{def:DCAI}).

\begin{remark}
It is worth to note on the interpretation of the dGLR in the context of arbitrage, which was first noticed in Bernardo and Ledoit~\cite{Bernardo2000} for the static Gain-Loss Ratio.
Observe that
\begin{align*}
\sum_{s=t}^T H_s(\omega) &\geq 0 \;  \text{for all }\;  \omega \in \Omega,  \quad  \bE^\bP_t\Big[\sum_{s=t}^T H_s\Big](\omega)>0\;  \text{for some }\; \omega \in \Omega
\end{align*}
  is equivalent to
\begin{align*}
\bE^\bP_t\Big[\big(\sum_{s=t}^T H_s\big)^-\Big](\omega) &= 0 \;  \text{for all }\;  \omega \in \Omega,  \quad  \bE^\bP_t\Big[\sum_{s=t}^T H_s\Big](\omega)>0\;  \text{for some }\; \omega \in \Omega,
\end{align*}
which is ultimately equivalent to
\begin{align*}
\textrm{dGLR}_t(H)(\omega)=\infty \quad \text{for some }\;  \omega \in \Omega.
\end{align*}
Therefore, in view of Definition~\ref{Def:ArbOpp}, a cash flow $H \in \cH(t)$ is an arbitrage opportunity at time $t \in \cT$ if and only if $\textrm{dGLR}_t(H)(\omega)=\infty$ for some $\omega \in \Omega$.
Equivalently, the no-arbitrage condition holds at time $t \in \cT$ if and only if $\textrm{dGLR}_t(H)$ is bounded for all $H \in \cH(t)$.
\end{remark}

In order to apply the general theory developed above, we will find the sets of probability measures that correspond to dGLR.
We define a family  $\widehat{\cQ}$  as
  \begin{align}\label{eq:QdGLR}
     \widehat{\cQ}^{\gamma} &:=  \Big\{ \bQ : \mathrm{d}\bQ/\mathrm{d}\bP= c(1+\Lambda), \; c > 0, \; \Lambda \in \mathfrak{L}^\gamma , \,c\,\bE^{\bP}[1+\Lambda] = 1   \Big\},
    \end{align}
for all $\gamma \in (0, \infty)$, where
\begin{align*}
   \mathfrak{L}^\gamma&:=\set{\Lambda \, :\, \Lambda \;\text{is an}\; \cF_T\textrm{-measurable r.v.}, \; 0\leq \Lambda \leq \gamma}.
\end{align*}
\begin{remark}\mbox{} \label{remark:AssumptionsDGLR}
\begin{itemize}
  \item[(i)]
For each $\gamma \in (0, \infty)$, the set of densities $\widehat{\cE}^{\gamma}$ defined as
\begin{equation*}
  \widehat{\cE}^{\gamma}  :=  \Big\{ \frac{\d \bQ}{\d \bP}: \bQ \in \cQ^{\gamma}_t\Big\}
\end{equation*}
is closed and convex. Thus, dGLR satisfies Assumption~B.
\item[(ii)]
For each  $t \in \cT, D \in  L^0$,  the function of $\gamma\in (0, \infty)$ defined as
\begin{equation}\label{DCRM}
\rho^{\gamma}_t(D):=\inf_{\bQ \in \widehat{\cQ}^{\gamma}} \bE^\bQ_t \Big[ \sum_{s=t}^T D_s \Big] \, ,
\end{equation}
 is continuous, and hence dGLR satisfies Assumption~C.
Indeed,  for each $\omega \in \Omega$ we have that
  \begin{align*}
    \inf_{\bQ \in \widehat{\cQ}^{\gamma}} \bE^\bQ_t\Big[ \sum_{s=t}^T D_s\Big](\omega)&=\inf_{\eta \in \widehat\cE^{\gamma}} \frac{\bE^\bP_t\big[ \eta \sum_{s=t}^T D_s\big](\omega)}{\cE^\bP_t[\eta](\omega)}
    =\inf_{\Lambda \in \mathfrak{L}^\gamma} \frac{\bE^\bP_t\big[ (1+\Lambda) \sum_{s=t}^T D_s\big](\omega)}{\bE^\bP_t[1+\Lambda](\omega)}.
  \end{align*}
  \item[(iii)]
Note that the LHS of \eqref{DCRM} is the value of a DCRM associated with $\widehat{\cQ}$ (see~\ref{theorem:DCRMrep}).
   \end{itemize}
\end{remark}
\noindent

\begin{proposition}\label{Prop:dGLR}
 The family $\widehat{\cQ}$, defined in \eqref{eq:QdGLR} is an increasing family of dynamically consistent sets of probability measures that corresponds to  \emph{dGLR}.
 \end{proposition}

\begin{proof} We start by observing that, for each $\gamma>0$, the set $\widehat{\cQ}^{\gamma}$ is nonempty since, in particular, we may take $\Lambda=0$ in the definition of $\widehat{\cQ}^\gamma$.
Clearly,  $\widehat{\cQ}^{\gamma}$ is increasing in $\gamma$.

For the rest of the proof we fix $\gamma >0$.  We denote by  $\Upsilon^t=\{P^t_1,P^t_2,\dots,P^t_{n_t}\}$ the unique partition of $\Omega$ at time $t$ that generates $\mathcal{F}_t$. In order to prove our result it suffices to show that $\widehat{\cQ}^{\gamma}$ is weakly consistent (see Corollary 4.1.1 in~\cite{ZhangPhDThesis2011}), which is
\begin{align}
\1_{P^t_i}  \inf_{\bQ\in\cQ^{\gamma}}\bE^{\bQ}_t[X]
&\leq \1_{P^t_i} \max_{\omega\in P^t_i}\bigg\{\inf_{\bQ\in\cQ^{\gamma}}\bE^{\bQ}_{t+1}[X](\omega)\bigg\}\,,
\end{align}
for every $t\in\{0,\dots,T-1\}$, $P^t_i \in \Upsilon^t$, and $X\in\cF_T$.
Next, take  $0 \leq \Lambda \leq \gamma$ and suppose that
\begin{equation*}
\max\limits_{\omega\in P^t_i}  \frac{\bE^\bP_{t+1}\big[(1+\Lambda)X\big](\omega)}{\bE^\bP_{t+1}[1+\Lambda](\omega)} \leq a,
\end{equation*}
for some $a \in \bR$.
Applying the tower property of conditional expectations, we deduce that the following implication holds:
\begin{equation*}
\max\limits_{\omega\in P^t_i}  \frac{\bE^\bP_{t+1}\big[(1+\Lambda)X\big](\omega)}{\bE^\bP_{t+1}[1+\Lambda](\omega)} \leq a \ \  \Rightarrow \ \ \max\limits_{\omega\in P^t_i}  \frac{\bE^\bP_{t}\big[(1+\Lambda)X\big](\omega)}{\bE^\bP_{t}[1+\Lambda](\omega)} \leq a.
\end{equation*}
Hence, since $a$ is arbitrary, we deduce that
\begin{align*}
\1_{P^t_i} \frac{\bE^\bP_{t}\big[(1+\Lambda)X\big](\omega)}{\bE^\bP_{t}[1+\Lambda](\omega)}\leq  \1_{P^t_i} \max_{\omega\in P^t_i} \frac{\bE^\bP_{t}\big[(1+\Lambda)X\big](\omega)}{\bE^\bP_{t}[1+\Lambda](\omega)} \leq \1_{P^t_i} \max_{\omega\in P^t_i} \frac{\bE^\bP_{t+1}\big[(1+\Lambda)X\big](\omega)}{\bE^\bP_{t+1}[1+\Lambda](\omega)}
\end{align*}
for all $\omega \in \Omega$.
Thus, for $\bQ = c(1+\Lambda)\bP$, we obtain
\begin{align*}
  \1_{P^t_i} \bE^\bQ_{t}[X](\omega) \leq \1_{P^t_i}\max_{\omega\in P^t_i} \bE^\bQ_{t+1}[X](\omega),
\end{align*}
for all $\omega \in \Omega$.
Therefore,
\begin{align*}
\1_{P^t_i} \inf_{\bQ\in\cQ^{\gamma}}\bE^{\bQ}_t[X]
&\leq  \1_{P^t_i}\inf_{\bQ \in\cQ^{\gamma}}\bigg\{ \max_{\omega\in P^t_i}\bE^{\bQ}_{t+1}[X](\omega)\bigg\} \leq \1_{P^t_i} \max_{\omega\in P^t_i}\bigg\{\inf_{\bQ \in\cQ^{\gamma}}\bE^{\bQ}_{t+1}[X](\omega)\bigg\},
\end{align*}
which proves the weak consistency of $\widehat{\cQ}^\gamma$.

We now show that the family $\widehat{\cQ}$ corresponds to the dGLR.
By Lemma~\ref{lemma:CorrFamilyQ}, this is equivalent to show that
\begin{equation}\label{eq:dGLR3}
  \mathrm{dGLR}_t(D)(\omega) \geq \gamma  \quad \Longleftrightarrow \quad  \inf_{ \mathbb{Q} \in \widehat{\cQ}^{\gamma} } \bE^{\mathbb{Q}}_t \big[ X_t^T \big](\omega)  \geq 0,
\end{equation}
for all $\omega \in \Omega$, $t \in \cT$ and $D \in L^0$, where for convenience we denoted  $X^T_t=\sum_{u=T}^{t}D_u$. In the rest of the proof we fix $\omega \in \Omega$, $t \in \cT$ and $D \in L^0$.

In order to show \eqref{eq:dGLR3}, we first observe that since any $\eta \in \cE^{\gamma}$ is strictly positive, we may apply the abstract Bayes formula to write
\begin{align*}
\inf_{ \mathbb{Q} \in \widehat{\cQ}^{\gamma} } \bE^{\mathbb{Q}}_t \big[ X_t^T \big](\omega)  \geq 0 \quad & \Longleftrightarrow \quad   \inf_{ \eta \in \cE^{\gamma} } \bE^{\bP}_t \big [ \eta X_t^T \big](\omega)  \geq 0. \label{eq:dGLR16}
\end{align*}
Using the definition of $\cE^{\gamma}$ we deduce that
\begin{equation*} \label{eq:dGLR8}
  \inf_{ \eta \in \cE^{\gamma} } \bE^{\bP}_t \big[\eta  X_t^T \big](\omega)=
  \inf_{\Lambda \in \mathfrak{L}^\gamma } \bE^{\bP}_t \big[ (1+\Lambda)X_t^T\big](\omega) =\bE^{\bP}_t\big[(1+\Lambda^*)X_t^T\big](\omega)
\end{equation*}
where $\Lambda^*:=\gamma \1_{\{X_t^T \leq 0\}}\in \mathfrak{L}^\gamma $.
As a result, it follows that
\begin{align*}
  \inf_{ \eta \in \cE^{\gamma} } \bE^{\bP}_t \big[\eta  X_t^T \big](\omega)&   =\bE^{\bP}_t\big[X_t^T\big](\omega)- \gamma \bE^{\bP}_t\big[\big(X_t^T\big)^-\big](\omega).
\end{align*}
Hence, we conclude that
\begin{equation*} \label{eq:dGLR15}
  \inf_{ \mathbb{Q} \in \widehat{\cQ}^{\gamma} } \bE^{\mathbb{Q}}_t \big[ X_t^T \big](\omega)  \geq 0 \quad
  \Longleftrightarrow \quad
\bE^{\bP}_t\big[X_t^T\big](\omega) \geq \gamma \bE^{\bP}_t\big[\big(X_t^T\big)^-\big](\omega).
\end{equation*}
By the definition of the dGLR, it is clear  that \eqref{eq:dGLR3} is fulfilled.
\end{proof}

\subsection{Applications}\label{subsec:Applications}

In this section, using a simple model for ask and bid prices of a stock, and choosing the dGLR as acceptability index, we compute the good-deal ask and bid prices of a European-style Asian option in a market with transaction costs.
We compare these good-deal prices with the no-arbitrage bounds.
Recall that $\widehat \cQ$, defined in \eqref{eq:QdGLR}, is a dynamically consistent family of sets of probability measures that corresponds to the dGLR.
We compute the ask and bid prices using the representation result in Theorem~\ref{Thm:RepSpotAskBid}.
No-arbitrage price bounds are calculated via using the lower and upper no-arbitrage bounds defined in Section~\ref{sec:arbANDgooddeals}.

We suppose that the bid price of the stock is given in Table~\ref{table:StockModelBid}.
\begin{table}[!h]
\caption{Bid price paths of the stock}
\smallskip
\centering
\renewcommand{\arraystretch}{1}
\begin{tabular}{cccccc}
     \hline \hline
      & $\omega_1$ & $\omega_2$ & $\omega_3$ & $\omega_4$ & $\omega_5$\\ \hline
  $t=0$ & $50$ & $50$ & $50$ & $50$ & $50$ \\
  $t=1$ & $80$ & $80$ & $80$ & $40$ & $40$ \\
  $t=2$ & $90$ & $70$ & $60$ & $60$ & $30$ \\\hline
\end{tabular}
\label{table:StockModelBid}
\end{table}
The ask price process is assumed to satisfy $P^{ask}:=P^{bid}(1+\lambda)$, where $\lambda \in \bR_+$ is the \emph{transaction costs coefficient}.
We also define the \emph{mid price process} as $P^{mid}:=(P^{ask}+P^{bid})/2$.

We recall that $\widehat \cQ$ is defined in terms of the reference measure $\bP$, which we will now assume to be
\begin{equation*}
\big(\bP(\omega_1), \bP(\omega_2), \bP(\omega_3), \bP(\omega_4), \bP(\omega_5)\big)=(1/10, 1/8, 1/4, 1/4, 11/40) \, .
\end{equation*}

\begin{example}[Asian Call Option]\label{subsubsec:ExampleA1dGLR}
We now compute the ask and bid price of a European-style Asian call option with a strike of 65.
According to our two-period model, the derivative contract is defined as
\begin{equation*}
D:= \left(0, 0,  \Big(( P^{mid}_0+P^{mid}_{1} +P^{mid}_2)/3- 65 \Big)^+\right).
\end{equation*}
Recall that $\Pi^{ask, \gamma}(D)$ and $\Pi^{bid, \gamma}(D)$ denote the good-deal prices computed using the dGLR, whereas $S^{ask}(D)$ and $S^{bid}(D)$ are the upper and lower no-arbitrage bounds, respectively.
Our results are presented in Table~\ref{table:ameanAsianprices} and Table~\ref{table:ameanAsianpricest1} for different transaction cost coefficients at $t=0,1$.
The prices displayed in Table~\ref{table:ameanAsianpricest1} correspond to the upper node of the tree, since the prices for the lower node are equal to zero.
In Figure~\ref{LiqSurfAA} we display the ``liquidity surface'', which is the plot of good-deal bid-ask spread as a function of the level $\gamma$ and transaction costs coefficient $\lambda$ at $t=0$.

\begin{table}[tH]
\caption{Ask and Bid Prices of an Arithmetic Asian Call Option, $t=0$}
\smallskip
\centering
\renewcommand{\arraystretch}{1.3}
\begin{tabular}{|c|c|c||c|c||c|c|}
\cline{2-7} \cline{2-7}
 \multicolumn{1}{c}{ }& \multicolumn{2}{||c||}{$\lambda = 0$}  & \multicolumn{2}{c||}{$\lambda = 0.005$} & \multicolumn{2}{c|}{$\lambda = 0.01$}  \\
\cline{2-7}
 \multicolumn{1}{c}{ }&  \multicolumn{1}{||c|}{$S^{ask}_0(D)$ }& $S^{bid}_0(D)$ & $S^{ask}_0(D)$ & $S^{bid}_0(D)$ & $S^{ask}_0(D)$ & $S^{bid}_0(D)$   \\
  \cline{2-7}
\multicolumn{1}{c}{ } & \multicolumn{1}{||c|}{1.38885 }  & 1.25003 & 1.48402 & 1.23020 & 1.55003 & 1.16726  \\ \hline
 \multicolumn{1}{|c}{\backslashbox{\vspace{-0.01in} \hspace{0.001cm} $\gamma$  }{}}
 &  \multicolumn{1}{||c|}{$\Pi^{ask, \gamma}_0(D)$} & $\Pi^{bid, \gamma}_0(D)$ & $\Pi^{ask, \gamma}_0(D)$ & $\Pi^{bid, \gamma}_0(D)$ & $\Pi^{ask, \gamma}_0(D)$ & $\Pi^{bid, \gamma}_0(D)$ \\ \hline
 0.0001 & \multicolumn{1}{||c|}{1.34177} & 1.34155 & 1.37681 & 1.37659 &  1.41186 & 1.41163  \\
 0.001 & \multicolumn{1}{||c|}{1.34274} & 1.34058 & 1.37781 & 1.37560 & 1.41288 & 1.41061   \\
 0.005 & \multicolumn{1}{||c|}{1.34706} & 1.33628 & 1.38224 & 1.37118 &  1.41742 & 1.40609 \\
 0.01 & \multicolumn{1}{||c|}{1.35244} & 1.33095 &  1.38776 & 1.36571 & 1.42309 & 1.40047  \\
 0.05 & \multicolumn{1}{||c|}{1.38885} & 1.28975 &  1.43158 & 1.32344 & 1.46802 & 1.35712   \\
 0.1 & \multicolumn{1}{||c|}{1.38885} & 1.25003 &  1.48402 & 1.27414 &  1.52322 & 1.30657   \\
 0.25 & \multicolumn{1}{||c|}{1.38885} & 1.25003 &  1.48402 & 1.23020 &  1.55003 & 1.17523  \\
 0.5  & \multicolumn{1}{||c|}{1.38885} & 1.25003 & 1.48402 & 1.23020 &  1.55003 & 1.16726   \\
 0.75  & \multicolumn{1}{||c|}{1.38885} & 1.25003 & 1.48402 & 1.23020 &  1.55003 & 1.16726  \\
 1   & \multicolumn{1}{||c|}{1.38885} & 1.25003 & 1.48402 & 1.23020 &   1.55003 & 1.16726   \\
 1.25 & \multicolumn{1}{||c|}{1.38885} & 1.25003 & 1.48402 & 1.23020 &  1.55003 & 1.16726   \\ \hline
\end{tabular}
\label{table:ameanAsianprices}
\end{table}

\begin{table}[tH]
\caption{Ask and Bid Prices of an Arithmetic Asian Call Option, $t=1$, $\omega=\omega_1$}
\smallskip
\centering
\renewcommand{\arraystretch}{1.3}
\begin{tabular}{|c|c|c||c|c||c|c|}
\cline{2-7} \cline{2-7}
 \multicolumn{1}{c}{ }& \multicolumn{2}{||c||}{$\lambda = 0$}  & \multicolumn{2}{c||}{$\lambda = 0.005$} & \multicolumn{2}{c|}{$\lambda = 0.01$}  \\
\cline{2-7}
 \multicolumn{1}{c}{ }&  \multicolumn{1}{||c|}{$S^{ask}_1(D)$ }& $S^{bid}_1(D)$ & $S^{ask}_1(D)$ & $S^{bid}_1(D)$ & $S^{ask}_1(D)$ & $S^{bid}_1(D)$   \\
  \cline{2-7}
\multicolumn{1}{c}{ } & \multicolumn{1}{||c|}{5.55541} & 5.00014 & 5.67765 & 5.17512 & 5.79988 & 5.35011  \\ \hline
 \multicolumn{1}{|c}{\backslashbox{\vspace{-0.01in} \hspace{0.001cm} $\gamma$  }{}}
 &  \multicolumn{1}{||c|}{$\Pi^{ask, \gamma}_1(D)$} & $\Pi^{bid, \gamma}_1(D)$ & $\Pi^{ask, \gamma}_1(D)$ & $\Pi^{bid, \gamma}_1(D)$ & $\Pi^{ask, \gamma}_1(D)$ & $\Pi^{bid, \gamma}_1(D)$ \\ \hline
 0.0001 & \multicolumn{1}{||c|}{5.36684} & 5.36648 & 5.50701 & 5.50665 & 5.64718 & 5.64681  \\
 0.001 &  \multicolumn{1}{||c|}{5.36847} & 5.36485 & 5.50866 & 5.50499 & 5.64886 & 5.64513   \\
 0.005 &  \multicolumn{1}{||c|}{5.37568} & 5.35763 & 5.51598 & 5.49767 & 5.65628 & 5.63770 \\
 0.01 &   \multicolumn{1}{||c|}{5.38465} & 5.34864 & 5.52508 & 5.48854 & 5.66551 & 5.62844  \\
 0.05 &   \multicolumn{1}{||c|}{5.45447} & 5.27791 & 5.59591 & 5.41678 & 5.73736 & 5.55566   \\
 0.1 &    \multicolumn{1}{||c|}{5.53722} & 5.19249 & 5.67764 & 5.33012 & 5.79988 & 5.46775   \\
 0.25 &   \multicolumn{1}{||c|}{5.55541} & 5.00014 & 5.67765 & 5.17512 & 5.79988 & 5.35012  \\
 0.5  &   \multicolumn{1}{||c|}{5.55541} & 5.00014 & 5.67765 & 5.17512 & 5.79988 & 5.35011  \\
 0.75  &  \multicolumn{1}{||c|}{5.55541} & 5.00014 & 5.67765 & 5.17512 & 5.79988 & 5.35011  \\
 1   &    \multicolumn{1}{||c|}{5.55541} & 5.00014 & 5.67765 & 5.17512 & 5.79988 & 5.35011  \\
 1.25 &   \multicolumn{1}{||c|}{5.55541} & 5.00014 & 5.67765 & 5.17512 & 5.79988 & 5.35011 \\ \hline
\end{tabular}
\label{table:ameanAsianpricest1}
\end{table}

\begin{figure}[!h] \vspace{0.01cm}
\centering
\begin{tabular}{c}
\epsfig{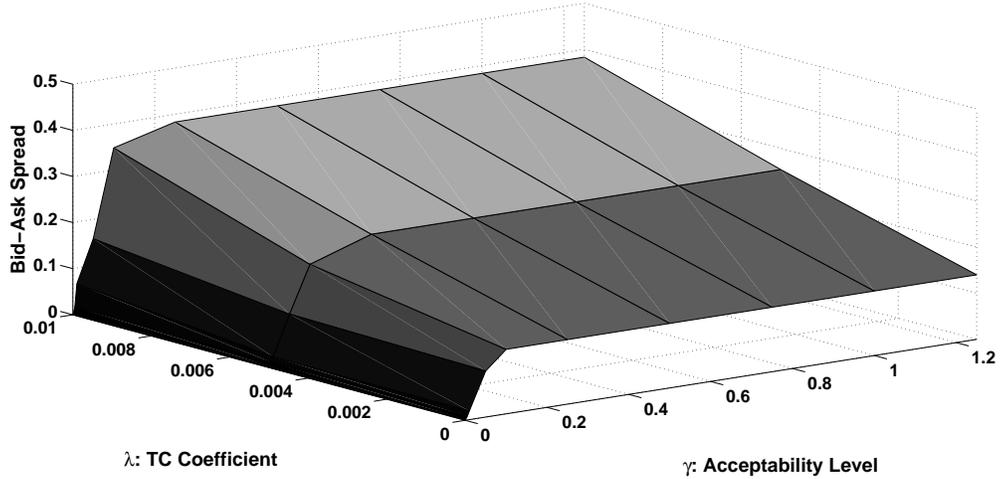}
\end{tabular}
\caption{Liquidity Surface of an Asian call Option at $t=0$}
\label{LiqSurfAA}
\end{figure}

In Figure~\ref{LiqSurfAA}, it is apparent that the good-deal bid-ask spread is increasing both in the acceptance level $\gamma$ and in the transaction cost coefficient $\lambda$.
The good-deal bid-ask spread naturally increases in  $\gamma$ because of the representations in Theorem~\ref{Thm:RepSpotAskBid}, and since $\cQ^{\gamma}$ is increasing in $\gamma$.
On the other hand, the good-deal bid-ask spread, as well as the difference between the upper and lower no-arbitrage bounds, increases in $\lambda$ since hedging the claim becomes more expensive as the $\lambda$ increases.

We also note from Table~\ref{table:ameanAsianprices} that both no-arbitrage bounds and the good-deal prices increase in $\lambda$, and that the good-deal ask and bid prices converge to the no-arbitrage bounds at higher $\gamma$ values.
This is also due to the fact that hedging is more expensive as $\lambda$ increases.
For example, in case $\lambda = 0$, $\Pi^{ask, \gamma}_0(D)$ and $\Pi^{bid, \gamma}_0(D)$ approximately converge to $S^{ask}_0(D)$ and $S^{bid}_0(D)$, respectively, at $\gamma = 0.1$, whereas if $\lambda = 0.005$ this happens at approximately $\gamma = 0.25$, and in the case $\lambda = 0.01$ it happens at approximately $\gamma = 0.5$.
\end{example}

\bigskip

\appendix

\section{Dynamic coherent acceptability indices}\label{sec:PricingDCAI}
{
\footnotesize
In this section, we provide some useful background information about acceptability indices and risk measures, as studied in \cite{ChernyMadan2009} and \cite{BCZ2010}.
Investors are usually concerned with finding satisfactory balance between \textit{reward and risk} associated with an investment process.
Various measures have been developed to quantify this balance.
Such measures are typically referred to as \textit{performance measures} or \textit{measures of performance}.
Cherny and Madan \cite{ChernyMadan2009} originated an effort to provide a mathematical framework to study these measures in a unified way for static models, and Bielecki et al. \cite{BCZ2010} followed up with an extension to a dynamic set-up.

A very popular measure of performance is the Sharpe Ratio introduced by \cite{Sharpe1964a}.
The Sharpe Ratio is expressed as the \textit{ratio} of expected excess return to standard deviation, and thus in financial applications it measures expected excess return of a portfolio in units of portfolio's standard deviation.
It has been used as a classical tool to rank portfolios according to their ``reward-to-risk" characteristics.

Using standard deviation to quantify risk is considered to be the major drawback of Sharpe Ratio because positive returns also contribute to this measure of risk.
To eliminate this unwanted feature other ratio-types performance measures that consider the downside risk were proposed, such as Sortino Ratio, \cite{SortinoPrice1994}, and Gain Loss Ratio (GLR), \cite{Bernardo2000}.
Another popular generalization of the Sharpe Ratio is provided by the Risk Adjusted Return on Capital, which is constructed as a  ratio of mean excess return to some selected measure of risk.

All the performance measures mentioned above share some common desirable features: they are unit-less, they are increasing functions of reward and decreasing functions of risk; moreover, according to these performance measures diversification of a portfolio improves its performance.
This observation prompts a natural desire to study performance measures in a unified mathematical framework. As already mentioned, such a study was recently originated by \cite{ChernyMadan2009}.
The study of \cite{ChernyMadan2009} was done in a static, one-time period setup, and the authors coined the term \emph{acceptability index} as a mathematical terminology for a performance measure.
In \cite{BCZ2010}, this static mathematical framework for studying acceptability indices was elevated to a dynamical, multi-period setup, where cash flows are considered as random processes and acceptability is assessed consistently in time.
In particular, they measure the performance of the total cumulative terminal value of the cashflow as seen from the initial time of the investment process, and also all remaining cumulative cashflows between each intermediate time and the terminal time of the investment process.

We proceed by recalling definitions and results from the theory of Dynamic Coherent Acceptability Indices, that were studied in Bielecki~et~al.~\cite{BCZ2010}.

We first recollect the definition of a dynamic coherent acceptability index.

\begin{definition}\label{def:DCAI}
A dynamic coherent acceptability index (DCAI) is a function
$\alpha: \cT\times L^0\times\Omega\to [0,\infty]$ that satisfies the following properties:
\begin{enumerate}
\item[\bf (D1)] {\bf Adaptiveness.}
For any $t\in\cT$ and $D\in L^0$, $\alpha_t(D)$ is $\mathcal{F}_t$-measurable;
\item[\bf (D2)] {\bf Independence of the past.}
For any $t\in\cT$ and $D, D'\in L^0$, if there exists $A\in\mathcal{F}_t$ such that $\1_A D_s=\1_A D'_s$ for
all $s\geq t$, then $\1_A\alpha_t(D)=\1_A\alpha_t(D')$;
\item[\bf (D3)] {\bf Monotonicity.}
For any $t\in\cT$ and $D, D'\in L^0$, if $D_s(\omega)\geq D'_s(\omega)$ for all $s\geq t$ and $\omega\in\Omega$, then $\alpha_t(D)\geq \alpha_t(D')$ for all $\omega\in\Omega$;
\item[\bf (D4)] {\bf Scale invariance.}
$\alpha_t(\lambda D) = \alpha_t(D)$ for
all  $\lambda >0, \ D\in L^0, \ t\in\cT,$ and $\omega\in\Omega$;
\item[\bf (D5)] {\bf Quasi-concavity.}
 If $\alpha_t(D) \geq x$ and $\alpha_t(D^\prime)\geq x$ for some $t\in\cT$, $\omega\in\Omega$,
 $D,D'\in L^0$, and $x\in (0,\infty]$, then
 $\alpha_t(\lambda D + (1-\lambda)D^\prime) \geq x$ for all $\lambda\in[0,1]$;
\item[\bf (D6)] {\bf Translation invariance.}
    $\alpha_t(D+m\1_{\{t\}})=\alpha_t(D+m\1_{\{s\}})$ for every
    $t\in\cT$, $D\in L^0$, $\omega\in\Omega$, $s\geq t$ and every $\mathcal{F}_t$-measurable random variable $m$;
\item[\bf (D7)] {\bf Dynamic consistency.} For any $t\in[0,\ldots, T-1]$ and $D,D'\in L^0$, if $D_t(\omega) \geq 0 \geq D'_t(\omega)$ for all $\omega\in\Omega$, and there exists a non-negative $\mathcal{F}_t$-measurable random variable $m$ such that $\alpha_{t+1}(D)\geq m(\omega)\geq \alpha_{t+1}(D')$ for all $\omega\in\Omega$, then $\alpha_t(D)\geq m(\omega)\geq \alpha_t(D')$ for all $\omega\in\Omega$.
\end{enumerate}
\end{definition}

\noindent Property (D1) is a natural property in a dynamic setup and it assumes that a DCAI is adapted to the same information flow $\{\mathcal{F}_t\}_{t\geq 0}$ as is any cash flow $D\in\mathcal{D}$. \smallskip

\noindent
 (D2) postulates that in the dynamic context the current measurement of performance of a cash flow $D$ only accounts for future payoffs. To decide, at any given point of time, whether one should hold on to a position generating the cash flow $D$, one may want to compare the measurement of the performance of the future payoffs (provided by DCAI at this  point of time) to already known past payoffs.
\smallskip

Properties (D3)-(D5) are naturally inherited from the static case.
\smallskip

\noindent
Translation invariance (D6) implies that  if a known dividend $m$ is added to $D$ at time $t$ (today), or at any future time $s\geq t$, then all such adjusted cashflows are accepted today at the same level.

\smallskip

\noindent
Dynamic consistency (D7) is the property in  the dynamic setup which relates the values of the index between two consecutive days in a consistent manner.
It can be interpreted from financial point of view as follows:
if a portfolio has a nonnegative cashflow today, then we accept this portfolio today at least at the same level as we would accept it tomorrow;
similarly, if the today's cashflow is nonpositive the acceptance level today can not be larger than the level of acceptance tomorrow.

For technical reasons, we assume that for every DCAI $\alpha$, and for every $t \in \cT$ and $\omega \in \Omega$, there exists two portfolios $D, D' \in \cD$ such that $\alpha_t(D)(\omega)=+\infty$ and $\alpha_t(D')(\omega)=0$.
In this case, we say that the DCAI $\alpha$ is \emph{normalized.
Assuming that $\alpha$ is normalized excludes degenerate examples of acceptability indices such as a constant index over all states, times, and portfolios.}

Let us proceed by stating with the definition of a dynamic coherent risk measure.
\begin{definition}\label{Def:DCRM}
Dynamic coherent risk measure (DCRM) is a function $\rho: \{0,\ldots,T\}\times L^0\times\Omega\to \mathbb{R}$ that satisfies the following properties:
\begin{enumerate}
\item[\bf (A1)] {\bf Adaptiveness.}
$\rho_t(D)$ is $\mathcal{F}_t$-measurable for all $t\in\cT$ and $D\in L^0$;
\item[\bf (A2)] {\bf Independence of the past.}
If  $\1_A D_s=\1_A D'_s$ for some $t\in\cT$, $D, D'\in L^0$, and $A\in\mathcal{F}_t$ and for all $s\geq t$, then $\1_A\rho_t(D)=\1_A\rho_t(D')$;
\item[\bf (A3)] {\bf Monotonicity.}
If $D_s(\omega)\geq D'_s(\omega)$ for some $t\in\cT$ and $D, D'\in L^0$,   and for all $s\geq t$ and $\omega\in\Omega$,
then $\rho_t(D)\leq \rho_t(D')$ for all $\omega\in\Omega$;
\item[\bf (A4)] {\bf Homogeneity.}
$\rho_t(\lambda D) = \lambda \rho_t(D)$ for
all  $\lambda > 0, \ D\in L^0, \ t\in\cT$, and $\omega\in\Omega$;
\item[\bf (A5)] {\bf Subadditivity.}
$\rho_t(D+D') \leq \rho_t(D) + \rho_t(D')$
for all $t\in\cT$, $D, D'\in L^0$, and $\omega\in\Omega$;
\item[\bf (A6)] {\bf Translation invariance.}
$\rho_t(D+m\1_{\{s\}})=\rho_t(D)-m$ for every
$t\in\cT$, $D\in L^0$, $\mathcal{F}_t$-measurable random variable $m$, and all $s\geq t$;
\item[\bf (A7)] {\bf Dynamic consistency.}
$$
\1_A (\min_{\omega\in A}\rho_{t+1}(D)-D_t) \leq \1_A \rho_t(D) \leq \1_A (\max_{\omega\in A}\rho_{t+1}(D)-D_t)\, ,
$$
for every $t\in\{0,1,\ldots, T-1\}$, $D\in L^0$ and $A\in\mathcal{F}_t$.
\end{enumerate}
\end{definition}

\noindent We want to mention that our definition of DCRM differs from the definition given in previous studies essentially only by the
dynamic consistency property.
For sake of completeness, we will present here how  property (A7) relates to other forms of dynamic
consistency of risk measures (for processes).
\begin{itemize}
\item[]{\bf (A7-I)} 
If $D_t=D'_t$, and $\rho_{t+1}(D)=\rho_{t+1}(D')$ for some
$t\in\{0,1,\ldots,T-1\}$, and $D,D'\in\mathcal{D}$,  then  $\rho_t(D)=\rho_t(D')$;
\item[]{\bf (A7-II)} 
$\rho_t(D)=\rho_t(-\rho_{t+1}(D)1_{\{t+1\}}) - D_t$ for
all times $t=0,1,\ldots,T-1$ and positions $D\in\mathcal{D}$.
\item[]{\bf(A7-III)} $\rho_t(D) \leq \rho_t(- \rho_{t+1}(D)1_{t+1}) - D_t$ for all $D\in\mathcal{D}, \ t\in\{0,1,\ldots,T-1\}$,
\item[]{\bf(A7-IV)} $\rho_t(D) \geq \rho_t(- \rho_{t+1}(D)1_{t+1}) - D_t$ for all $D\in\mathcal{D}, \ t\in\{0,1,\ldots,T-1\}$,
\item[]{\bf(A7-V)} if $D_t=0$, and $\rho_{t+1}(D) \leq 0$ for some $t\in\{0,1,\ldots,\}$ and $D\in\mathcal{D}$, then $\rho_t(D)\leq 0$.
\end{itemize}

Property (A7-I) is the dynamic consistency property for DCRM defined by  \cite{Riedel2004}.
Property (A7-II) is the version of the dynamic programming principle  (also called recursiveness),  introduced by \cite{CheriditoDelbaenKupper2006}.
Properties (A7-I) and (A7-II) are equivalent, and they are also sometimes called {\it strong dynamic consistency property}.
To the best of our knowledge, properties (A7-III) and (A7-IV) were first introduced in the context of random processes by \cite{AcciaioFollmerPenner2010}, and they were called {\it acceptance and rejection consistency}, respectively.
In the same paper, Acciaio, F\"{o}llmer and Penner introduced condition (A7-V) and they called it \textit{weakly acceptance consistent}.

It is straightforward to show that the dynamic consistency condition (A7) is stronger than (A7-V), and it is weaker than (A7-I) or (A7-II).
Also note that since conditions (A7-II) and (A7-III) taken together are equivalent to (A7-II), then, taken together they imply (A7).
However, the inverse implication is not necessarily true.

We now recall an important result that provides the representation of a DCAI in terms of a family of DCRMs, and the representation of DCRM in terms of a DCAI.
The proof the following theorem can be found in~\cite{BCZ2010}.

\begin{theorem} \mbox{} \label{theorem:Duality}
\begin{itemize}
\item[(i)]
If $\alpha$ is a normalized, right-continuous, dynamic coherent acceptability index,
then there exists a left-continuous and increasing family of dynamic coherent risk measures \\ $(\rho^{\gamma})_{{\gamma}\in(0,\infty)}$, such that
\begin{equation}
\alpha_t(D)( \omega)=\sup\{{\gamma}\in(0,\infty): \rho^{\gamma}_t(D)(\omega)\leq 0\},\, \qquad  \omega \in \Omega,\; t \in \cT, \; D\in L^0 .
\end{equation}
\item[(ii)] If $(\rho^{\gamma})_{{\gamma}\in(0,\infty)}$ is a left-continuous and increasing family of dynamic coherent risk measures,
 then there exists a right-continuous and normalized dynamic coherent acceptability index $\alpha$ such that,
\begin{align*}
\rho^{\gamma}_t(D)(\omega)=\inf\{c \in\mathbb{R}:\alpha_t(D+\delta_t(1c))(\omega)\geq {\gamma}\},\, \qquad \;\omega \in \Omega,\; t \in \cT, \; D\in  L^0 .
\end{align*}
\end{itemize}
We take $\inf\emptyset=\infty$ and $\sup\emptyset=0$.
\end{theorem}

Next, we recall the definitions of a dynamically consistent sequence of sets of probability measures and an increasing family of sequences of sets of probability measures.

\begin{definition}\mbox{}\label{Def:ConsistentProbSet}
\begin{itemize}
\item[(i)]
A sequence of sets of probability measures $(\cQ_t)_{t=0}^T$ absolutely continuous with respect to $\bP$ is called dynamically consistent with respect to the filtration $(\cF_t)_{t=0}^T$ if the sequence is of full-support and the following inequality holds
\begin{align*}
  \1_E \min_{\omega \in E} \Big\{\inf_{ \bQ \in \cQ_{t+1}} \bE^{\bQ}_{t+1}[X](\omega)\Big\} &\leq \1_E \inf_{ \bQ \in \cQ_t} \bE^{\bQ}_{t}[X] \leq   \1_E \max_{\omega \in E} \Big\{\inf_{ \bQ \in \cQ_{t+1}} \bE^{\bQ}_{t+1}[X](\omega)\Big\}
\end{align*}
for all $t \in \{0, 1, \dots, T-1\}$, $E \in \cF_t$, and $\cF_T$-measurable random variables $X$.
\item[(ii)]
A family of sequences of sets of probability measures
$((\cQ^{\gamma}_t)_{t=0}^T)_{\gamma \in(0,\infty)}$ is called increasing if $\cQ_t^{\gamma}\supseteq \cQ_t^{\beta}$, for all $\gamma \geq \beta>0$ and $t\in\cT$.
\end{itemize}
\end{definition}

Now, we recall a representation theorem for dynamic coherent risk measures in terms of dynamically consistent set of probabilities.
These results, combined with the results from Theorem~\ref{theorem:Duality} about duality between DCAI and DCRM, gives a representation theorem for dynamic coherent acceptability indices.

\begin{theorem}[Robust Representation Theorem for DCRM]\label{theorem:DCRMrep}
For $\gamma > 0$, a function $\rho^{\gamma}: \{0,1,\ldots,T\}\times L^0\times\Omega\to \mathbb{R}$ is a dynamic coherent risk measure
if and only if there exists a dynamically consistent family of sets of probabilities $(\cQ^{\gamma}_t)^T_{t=0}$ such that,
\begin{equation}\label{eq:RM-Representation}
\rho^{\gamma}_t(D)=-\inf_{\mathbb{Q}\in\cQ^{\gamma}_t}\mathbb{E}^{\mathbb{Q}}_t\Big[\sum_{s=t}^T D_s\Big]\, , \quad
 t\in\cT, \ D\in L^0.
\end{equation}
\end{theorem}
The proof this theorem can be found in~\cite{BCZ2010}.

A direct consequence of Theorem~\ref{theorem:Duality} and Theorem~\ref{theorem:DCRMrep}, is the following result, which is proved in~\cite{BCZ2010}.

\begin{theorem} \mbox{} \label{theorem:DCAIrep}
\begin{itemize}
\item[(i)]
Assume that $(\cQ^{\gamma}_t)_{t=0}^T)_{\gamma \in(0,\infty)}$ is an increasing family of dynamically consistent sequences of sets of probability measures.
Then, the function $\alpha: \{0,1,\ldots,T\}\times L^0\times\Omega\to[0,\infty]$ defined as follows,
\begin{equation*}
\alpha_t(D)(\omega)=\sup\Big\{\gamma \in(0,\infty): \inf_{\bQ\in\cQ^{\gamma}_t}\bE^{\bQ}_t\Big[\sum_{s=t}^T D_s\Big](\omega)\geq 0\Big\} \,,
\quad \omega \in \Omega, \;t\in\cT, \ D\in L^0,
\end{equation*}
is a normalized and right-continuous dynamic coherent acceptability index.
\item[(ii)]
If $\alpha$ is a normalized and right-continuous dynamic coherent acceptability index,
then there exists a family of dynamically consistent sequences of sets of probability measures
$(\cQ^{\gamma}_t)_{t=0}^T)_{\gamma \in(0,\infty)}$ such that
\begin{equation*}
\alpha_t(D)(\omega)=\sup\Big\{\gamma \in(0,\infty): \inf_{\mathbb{Q}\in\cQ^{\gamma}_t}\bE^{\bQ}_t\big[\sum_{s=t}^T D_s\big](\omega)\geq 0\Big\} \,,
\quad \omega\in \Omega, \;t\in\cT, \ D\in L^0.
\end{equation*}
Here we adopt the usual convention that  $\inf\emptyset=\infty$ and $\sup\emptyset=0$.
\end{itemize}
\end{theorem}

\section{Technical results}

The following lemma is an auxiliary result needed for Theorem~\ref{Thm:RepSpotAskBid}.

\begin{lemma}\label{lemma:ContinuityRhoSet}
  For any monotone increasing, continuous function $f:(0, \infty) \to \bR$, we have that
  \begin{equation*}
  f(\gamma) \leq 0 \quad \text{if and only if} \quad  \sup\{ \beta \in (0, \infty): f(\beta) \leq 0 \} \geq \gamma,
  \end{equation*}
 for any $\gamma>0$.
\end{lemma}

\begin{proof}
Let us define the set $\Gamma:=\{ \beta \in (0, \infty): f(\beta) \leq 0 \}$.
Assume that $f(\gamma) \leq 0$ for some $\gamma>0$. Then, $\gamma \in \Gamma$, and therefore $\sup \Gamma \geq \gamma$.

Conversely.  Suppose that $\sup \Gamma \geq \gamma$ and define $\beta^*:=\sup \Gamma$.
\noindent
If $\sup \Gamma =\infty$, then $f(x)\leq 0,$ for all $x>0$, and in particular for $x=\gamma$.
Now assume that $\beta^* \in (0, \infty)$.
\noindent
We first argue by contradiction that $\beta^* \in \Gamma$.
If $\beta^* \notin \Gamma$, then $f(\beta^*)>0$.
Now, since $f$ is continuous, there exists $\epsilon'>0$ so that $0<f(\beta^*-\epsilon')$.
By the definition of the supremum of a set, we have that, for all $\epsilon>0$, there exists $\beta^{\epsilon} \in \Gamma$ so that $\beta^*-\epsilon<\beta^{\epsilon}$.
Therefore, because $f$ is monotonically increasing, $f(\beta^* -\epsilon)\leq f(\beta^{\epsilon})$.
Hence, $0<f(\beta^* -\epsilon')\leq f(\beta^{\epsilon})$, which contradicts $\beta^{\epsilon} \in \Gamma$.
\noindent
We proceed by showing that $f(\gamma) \leq 0$.
Since $\gamma \leq \beta^*$ and $f$ is monotonically increasing, we have that $f(\gamma) \leq f(\beta^*)$.
However, $\beta^* \in \Gamma$, so $f(\gamma) \leq f(\beta^*)\leq 0$.
\end{proof}

We now recall a well-known characterization of compact sets.
For a proof, see Lemma I.5.6 in Dunford and Schwartz~\cite{DunfordSchwartz58}.

\begin{lemma}\label{thm:FiniteIntersection}
A subset of a topological space is compact if and only if every family of closed sets with the finite intersection property has a nonempty intersection.
\end{lemma}

The following theorem is an application of Hahn-Banach theorem, regarding the separation of hyperplanes.
\begin{theorem}\label{th:AppendHahnBanach}
If $\cZ$ and $\cC$ are disjoint closed convex subsets of $\bR^N$, and if $\cZ$ is compact, then there exists a constant $\epsilon$ with $\epsilon>0$, and a continuous linear functional $\varphi \in \bR^N$, so that
\begin{equation*}
  \varphi(c) \leq 0 < \epsilon < \varphi(z)
\end{equation*}
for all $z \in \cZ$ and $c \in \cC$.
\end{theorem}

\begin{proof}
By Theorem V.2.10 in Dunford and Schwartz~\cite{DunfordSchwartz58}, there exists constants $a$ and $\epsilon'$ with $\epsilon'>0$, and a continuous linear functional $\varphi \in \bR^N$, so that
\begin{equation}\label{eq:AppendHB}
  \varphi(x) \leq a- \epsilon' < a \leq \varphi(z)
\end{equation}
for all $z \in \cZ$ and $x \in \cC$.
We now argue that $\varphi(x) \leq 0$ for all $x \in \cC$.
Suppose there exists $a_0>0$ and $x_0 \in \cC$ so that $\varphi(x_0) =
a_0$.
Since $\cC$ is a cone, we have that $\lambda x_0 \in \cC$ for all $\lambda>0$. Thus,
\begin{equation*}
  \sup_{x \in \cC}\varphi(x) \geq \sup_{\lambda > 0} \varphi(\lambda x_0) = \sup_{\lambda>0} \lambda a_0 = + \infty,
\end{equation*}
which contradicts \eqref{eq:AppendHB}, and hence $\varphi(x)\leq0, \ x\in \cC$.
From here, and since $\varphi$ is linear and  $0\in\cC$, it follows that $\sup_{x\in\cC} \varphi(x)=0$. Thus, $a-\epsilon' \geq 0$, and hence $a>0$. Taking $\epsilon=a$ concludes the proof.
\end{proof}
}

\bibliographystyle{alpha}

\end{document}